\documentclass[11pt]{article}
\usepackage[left=1in,top=1in,right=1in,bottom=1in,head=.1in,nofoot]{geometry}

\usepackage{setspace,url,bm,amsmath} %

\usepackage{titlesec} %
\titlelabel{\thetitle.\quad} %

\titleformat*{\section}{\bf\Large\center}

\usepackage{graphicx} %
\usepackage{bbm}
\usepackage{latexsym}
\usepackage{caption}
\usepackage[margin=20pt]{subcaption}
\usepackage{hyperref}
\usepackage{booktabs}
\usepackage{multirow}

\usepackage{enumitem}

\usepackage[table]{xcolor}

\newcommand{\GG}[1]{}

\usepackage{amsthm}
\usepackage{amssymb}
\usepackage{amsmath}
\usepackage{color}

\usepackage{comment}
\theoremstyle{definition}

\newtheorem*{theorem*}{Theorem}
\newtheorem{theorem}{Theorem}
\newtheorem*{rmk*}{Remark}

\newtheorem{lemma}{Lemma}

\newtheorem{remark}{Remark}
\newtheorem{corollary}{Corollary}
\newtheorem*{corollary*}{Corollary}

\def\rank{\text{r}}

\usepackage{natbib} %
\bibpunct{(}{)}{;}{a}{}{,} %

\usepackage{etoolbox} %
\apptocmd{\sloppy}{\hbadness 10000\relax}{}{} %

\usepackage{color}
\usepackage{listings}

\def\Pr{\mathbb{P}}

\def\Var{\text{Var}}

\def\I{\mathbbm{1}}
\def\E{\mathbb{E}}

\def\bs{\boldsymbol}

\def\treat{\textup{1}}

\def\control{\textup{0}}

\usepackage[sc]{mathpazo}
\RequirePackage[normalem]{ulem}

\usepackage{soul}

\usepackage{accents}
\newcommand{\dtilde}[1]{\tilde{\raisebox{0pt}[0.85\height]{$\tilde{#1}$}}}

\usepackage{subcaption}
\usepackage[textsize=tiny, textwidth = 2cm, shadow]{todonotes}

\begin{document}

\onehalfspacing

\title{\bf 
Randomization Tests for Distributions of Individual Treatment Effects via Combined Rank Statistics
}
\author{
David Kim, Yongchang Su, Jake Bowers, and Xinran Li
\footnote{
David Kim is a Doctoral Candidate, Department of Statistics, University of Illinois 
at Urbana-Champaign, Champaign, IL 61820 (e-mail: 
\href{mailto:davidk9@illinois.edu}{davidk9@illinois.edu}). Yongchang Su recently received his PhD in Statistics from the University of Illinois at Urbana-Champaign (\href{mailto:ysu17@illinois.edu}{ysu17@illinois.edu}). 
Jake Bowers is Professor, Department of Political Science and Department of Statistics, University of Illinois at Urbana-Champaign, Champaign, IL 61820 (e-mail: \href{mailto:jwbowers@illinois.edu}{jwbowers@illinois.edu}). 
Xinran Li is Assistant Professor, Department of Statistics, University of Chicago, Chicago, IL, USA (e-mail: \href{mailto:xinranli@uchicago.edu}{xinranli@uchicago.edu}). 
Li was partly supported by
the U.S. National Science Foundation (DMS-2400961).
}
}
\date{}
\maketitle

\allowdisplaybreaks

\begin{abstract}

What proportion of treated units actually benefited from an experimental intervention? What is the median or the largest individual treatment effect? This paper develops methods for answering such questions about the distribution of individual causal effects in randomized experiments. Existing approaches require the analyst to select a rank-based test statistic before observing the data. A poor choice can substantially reduce power, while searching over multiple test statistics and adjusting for multiplicity using Bonferroni correction also incurs power loss. We propose inference procedures that adaptively combine multiple rank-based statistics while maintaining finite-sample validity. For stratified experiments, we further develop weighting schemes that effectively aggregate evidence across strata of heterogeneous sizes. The resulting combined test achieves power comparable to, or exceeding, that of the best individual test, without requiring prior knowledge of the optimal statistic. When applied to a randomized experiment evaluating a teacher training program, the combined test suggests that roughly half of treated teachers benefited, whereas a single rank-based test may indicate only a small minority. Thus, the choice of test determined whether the program appears broadly successful or narrowly effective.
\end{abstract}


\section{Introduction}

\subsection{Heterogeneous treatment effects and randomization-based inference}

What proportion of treated units actually benefited from an experimental intervention? When outcomes are unbounded, the average treatment effect does not indicate whether only a few benefited or nearly every unit did. A single average is compatible with many distributions of individual treatment effects---each calling for a different policy response.
Recently, two lines of research have developed to understand this heterogeneity. One prominent line of research focuses on inference for subgroup or conditional average treatment effects defined with respect to observed covariates \citep[see, e.g.,][]{imai2013, tian2014, NieWager2020, qu2024randomization}, as well as on optimal treatment rules that map observed covariates to treatment assignments to maximize overall outcomes in a population \citep[see, e.g.,][]{manski2004, zhang2012, athey2021observational}. A second line, to which we contribute,  
focuses on the distribution of individual treatment effects \citep{heckman1997}.  
Characterizing the full distribution---or its quantiles---allows researchers to address questions such as: What proportion of individuals benefit from, or are harmed by, the treatment? How large are the effects among the most or least affected subgroups?

We focus here on inference for distributions and quantiles of individual treatment effects. We adopt the randomization-based inference framework, which views potential outcomes as fixed constants and uses randomness solely from the treatment assignment as the ``reasoned basis'' \citep{Fisher:1935, Neyman:1923}.\footnote{Note that we focus on quantiles of individual treatment effects, rather than quantile treatment effects. The latter compare, for example, the $p$-th quantile of the potential outcome under treatment with the $p$-th quantile of the potential outcome under control \citep{koenker2017quantile}.}
Randomization-based inference, sometimes also called design-based or finite population inference, requires no distributional assumptions---neither independence nor identical distributions across units---and relies only on the experimental design.
Our work builds on recent randomization-based inference for treatment effect distributions and quantiles. 
Specifically, \citet{CDLM21quantile} proposed randomization tests with rank-based test statistics for inference on quantiles of individual treatment effects in completely randomized experiments. \citet{SL22quantile} extended that work to stratified randomized experiments and to sensitivity analyses in matched observational studies. \citet{ZL24quantile} reinterpreted these methods as constructing prediction intervals for effect quantiles within treated and control groups, 
which can then be combined to yield improved confidence intervals for effect quantiles across all units. 

\subsection{Our contribution}

The power of existing randomization tests for treatment effect quantiles depends on the choice of rank statistics. For example, the Stephenson rank sum statistic places greater weight on larger ranks through a tuning parameter \citep{Stephenson85, Rosenbaum2007dramatic} and can increase power when potential outcomes under treatment exhibit heavier right tails than those under control. \citet{CDLM21quantile} show that such statistics can yield more informative confidence bounds compared to the classical Wilcoxon rank sum statistic. However, that paper also presents situations where unlucky choices yield less powerful results than the Wilcoxon rank sum test. Fundamentally, the optimal choice of the rank statistic depends on the underlying unknown distributions of potential outcomes. Selecting and reporting the most favorable statistic after examining multiple options may therefore lead to substantial size distortion due to multiple testing.

To address this limitation, we develop inference procedures that adaptively combine multiple rank statistics in a data-driven manner, while preserving valid type I error control. Specifically, to test null hypotheses on treatment effect quantiles, we propose a test statistic that maximizes over several appropriately standardized rank statistics, with the standardization depending on whether the design is completely randomized or stratified. Combining multiple rank statistics in this way provides researchers a high-powered test without the burden of tuning parameter choice: as we show below, the new test provides confidence bounds that are more informative than, or nearly as informative as, those from the best choice of rank statistic across all quantiles of the treatment effect distribution. In an application to the randomized teacher training intervention in \citet{HSMHSD10}, we show that unlucky choice of rank test would suggest relatively few teachers benefited, whereas the combined rank test suggests that roughly half did.

In stratified experiments with varying stratum sizes, naive aggregation of rank statistics across strata may effectively ignore smaller strata---the Stephenson rank sum statistic can differ by orders of magnitude across strata of different sizes. We draw on the literature on classical nonparametric two-sample tests, such as \citet{van1960combination} and \citet{puri1965combination}, and propose normalized rank transformations and principled weighting schemes to balance contributions across strata. We also investigate two natural approaches to combining strata-specific rank statistics--—combining them before or after aggregation across strata---and compare their performance. 
We formulate the computations in these approaches as integer linear programming problems, which can often be solved efficiently.

The remainder of this paper is organized as follows. 
Section \ref{sec:framework} introduces the framework and notation. Section \ref{sec:CRE} develops adaptive inference procedures for quantiles of individual treatment effects in completely randomized experiments. Section \ref{sec:SRE} extends these to stratified randomized experiments and considers additionally weighting strategies for aggregating test statistics across strata. 
Section \ref{sec:simu} summarizes the main simulation findings. 
Section \ref{sec:app} applies the proposed methods to an empirical study of an experimental professional development program. Section \ref{sec:discussion} concludes with a brief discussion. All proofs of theorems and simulation studies are relegated to the supplementary material.

\section{Framework and Notation}\label{sec:framework}

\subsection{Potential outcomes and treatment assignment}

We consider an experiment with $n$ units and two treatment arms (labeled as
treatment and control). For each unit $1\le i\le n$, let $Y_i(1)$ and $Y_i(0)$
denote the treatment ($Z_i=1$) and control ($Z_i=0$) potential outcomes \citep{Neyman:1923,
Rubin:1974}, where $Z_i$ is the treatment assignment. The observed outcome $Y_i$ is one of the two potential outcomes depending on the value of the treatment assignment: $Y_i = Z_i Y_i(1) +
(1-Z_i) Y_i(0)$. The additive treatment effect for unit
$i$ is defined as $\tau_i = Y_i(1) - Y_i(0)$. For
descriptive convenience, let $\bs{Y}(1) = (Y_1(1), Y_2(1), \ldots,
Y_n(1))^\top$, $\bs{Y}(0) = (Y_1(0), Y_2(0), \ldots, Y_n(0))^\top$, $\bs{Y} =
(Y_1, Y_2, \ldots, Y_n)^\top$, $\bs{\tau} = (\tau_1, \tau_2, \ldots,
\tau_n)^\top$ and $\bs{Z} = (Z_1, Z_2, \ldots, Z_n)^\top$ denote the vectors of
treatment potential outcomes, control potential outcomes, observed outcomes,
individual treatment effects, and treatment assignments, respectively.

Throughout the paper, our statistical inference is randomization-based,  where
all potential outcomes $Y_i(1)$s and $Y_i(0)$s are viewed as fixed
constants or equivalently being conditioned on.
In particular, we regard the treatment assignments $Z_i$s as the sole source of randomness for statistical inference on the causal effects.

We focus on two common classes of experimental designs: the
completely randomized experiment (CRE) and the more general stratified
randomized experiment (SRE).
Under the CRE, 
the treatment assignment vector $\bs{Z}$ takes a particular value
$\bs{z} = (z_1, z_2, \ldots, z_n)\in \{0,1\}^n$ with probability $n_\treat!
n_\control!/n!$ if $\sum_{i=1}^n z_i = n_\treat$ and zero otherwise, where
$n_\treat$ and $n_\control \equiv n - n_\treat$ are some fixed positive
integers. We defer the discussion of the SRE to Section \ref{sec:SRE}.

\subsection{Treatment effect distribution and quantiles}

We are interested in inferring the distribution of individual treatment effects
for the $n$ units in the experimental pool, where we write the cumulative
distribution function of the $\tau_i$s as $F_n(c) = n^{-1} \sum_{i=1}^n
\I(\tau_i \le c)$ for $c\in \mathbb{R}$. 
Let $F_n^{-1}(\beta)$  for $0\le \beta \le 1$ be the corresponding treatment effect
quantile function. 
Since we are making inference for a finite set of $n$ units, the treatment effect quantile function can only take values from
the sorted individual treatment effects $\tau_{(1)} \le \tau_{(2)} \le \ldots
\le \tau_{(n)}$. Specifically, $F_n^{-1}(\beta) = \tau_{(k)}$ with $k=\lceil
n\beta \rceil$, the ceiling of $n\beta$. Therefore, inference about
$F_n(c)$ and $F_n^{-1}(\beta)$ is equivalent to inference about the sorted individual treatment effect
$\tau_{(k)}$ for $1\le k \le n$. Later in the paper, we will construct simultaneous
prediction intervals for all the $\tau_{(k)}$s, which in turn yield
confidence bands for the entire distribution function $F_n(\cdot)$ and quantile
function $F_n^{-1}(\cdot)$ of the individual treatment effects.

We tackle inference for treatment effect quantiles by considering the following classes of null hypotheses on treatment effects among treated and control units, respectively:
\begin{align}\label{eq:H_kc_t}
    H_{k, c}^\treat:& \tau_{\treat(k)}\le c \Longleftrightarrow n_{\treat}(c) \le n_\treat-k  \Longleftrightarrow \bs{\tau} \in \mathcal{H}_{k,c}^\treat, \quad (1\le k \le n_\treat, c\in \mathbb{R})
    \\
    \label{eq:H_kc_c}
    H_{k, c}^\control:& \tau_{\control(k)}\le c \Longleftrightarrow n_{\control}(c) \le n_\control-k \Longleftrightarrow \bs{\tau} \in \mathcal{H}_{k,c}^\control, \quad (1\le k \le n_\control, c\in \mathbb{R})
\end{align}
where $\tau_{\treat(1)} \le \tau_{\treat(2)} \le \ldots \le \tau_{\treat(n_\treat)}$ are the sorted individual treatment effects for the treated units, $n_{\treat}(c) = \sum_{i=1}^n Z_i \I(\tau_i > c)$ denotes the number of treated units with individual effects exceeding $c$, and $\mathcal{H}_{k,c}^\treat \equiv \{\bs{\delta}\in \mathbb{R}^n: \sum_{i=1}^n Z_i \I(\delta_i > c) \le n_\treat-k\}$ denotes the set of all possible treatment effect vectors under the null hypothesis focusing on the treated units. 
We define $\tau_{\control(k)}$, $n_{\control}(c)$ and  $\mathcal{H}_{k,c}^\control$ analogously for control units.
The double arrows in \eqref{eq:H_kc_t} and \eqref{eq:H_kc_c} connect three equivalent ways of stating each null hypothesis: a claim about the individual effect at rank $k$, a count of how many individual effects exceed a threshold, and membership in a set of possible effect vectors.

In other words, the null hypothesis $H_{k,c}^\treat$ states: ``The $k$th smallest individual treatment effect among treated units is at most $c$.'' If we sort the 164 treated teachers in the \citet{HSMHSD10} study by how much each one benefited, then $H_{k,c}^\treat$ claims that teacher at rank $k$ gained no more than $c$ points on the outcome due to the training.

We make a few remarks. First, the null hypotheses in \eqref{eq:H_kc_t} and \eqref{eq:H_kc_c} are themselves random, depending on the observed treatment assignment $\bs{Z}$. This differs from traditional fixed null hypotheses. 
We say a $p$-value is valid for a random null hypothesis $H$, if and only if $\Pr(p \le \alpha \text{ and } H \text{ holds})\le \alpha$ for $\alpha \in (0,1)$. 
Importantly, inverting such a valid $p$-value can yield valid {\it prediction} sets for certain {\it random} quantities, such as $\tau_{\treat(k)}$s and $\tau_{\control(k)}$s, specified by the null hypotheses.

Second, once we have valid tests for the null hypotheses in \eqref{eq:H_kc_t}, by switching the treatment labels and changing the outcome signs, we can obtain valid tests for the null hypotheses in \eqref{eq:H_kc_c} for treatment effects among control units. 
Therefore, without loss of generality, we will focus on the hypotheses in \eqref{eq:H_kc_t}. 
Moreover, by only changing the outcome signs, we can test null hypotheses as in \eqref{eq:H_kc_t} and \eqref{eq:H_kc_c} but with alternative hypotheses favoring smaller treatment effects; see Section \ref{sec:discussion} for further discussion. 

Third, we will invert the tests for hypotheses in \eqref{eq:H_kc_t} and \eqref{eq:H_kc_c} to construct simultaneous prediction intervals for treatment effect quantiles among treated and control units, respectively. 
With a Bonferroni correction, these can be further combined to provide simultaneous confidence intervals for treatment effect quantiles $\tau_{(k)}$s among all units.

\subsection{Existing approaches and their limitations}\label{subsec:limitation}

\cite{CDLM21quantile} and \cite{SL22quantile} recently proposed valid tests for the following null hypotheses on quantiles of individual treatment effects under the CRE and SRE, respectively:
\begin{align}\label{eq:H_kc}
    H_{k, c}: \tau_{(k)}\le c \Longleftrightarrow n(c) \le n-k, 
    \quad (1\le k \le n, c\in \mathbb{R}),
\end{align}
where $n(c)$ denotes the number of units whose individual effects exceed $c$.
$H_{k,c}$ states that the $k$th smallest individual treatment effect among \emph{all} units
 does not exceed $c$.
 \citet{ZL24quantile} showed that $p$-values for $H_{k, c}$ in \eqref{eq:H_kc}
 are also valid for $H_{k-n_\control, c}^\treat$ in \eqref{eq:H_kc_t} or
 $H_{k-n_\treat, c}^\control$ in \eqref{eq:H_kc_c}. 
 This enables a
 divide-and-combine strategy: construct prediction intervals for effect
 quantiles separately among treated and among control units, then combine them
 (with a Bonferroni correction) to obtain confidence intervals for effect
 quantiles among all units.

The approaches in \cite{CDLM21quantile}, \cite{SL22quantile} and
\citet{ZL24quantile} work for a general class of (stratified) rank sum
statistics with various rank transformations. For example, the Stephenson rank
statistics \citep{Stephenson85} can produce more powerful tests for \eqref{eq:H_kc}
(as well as \eqref{eq:H_kc_t} and \eqref{eq:H_kc_c}) than the commonly used
Wilcoxon rank statistic, especially when the treatment potential outcomes tend
to have heavier right tails than the control ones.\footnote{This is why
\citet{rosenbaum2008aberrant} use these tests to detect "aberrant" or extreme treatment
effects.} However, the ``optimal'' choice of rank
statistics generally depends on the underlying true data-generating
process, which depend on unobserved potential outcomes. Furthermore, even under
the same data-generating process, we may prefer different rank statistics
when inferring different quantiles of the individual treatment effects. 
For example, a test statistic that has good sensitivity for detecting maximum effects may have poor sensitivity for detecting median effects. Our work here aims to
side-step these problems by choosing the test statistics adaptively based on
the observed data and the quantile of interest, while still maintaining 
type-I error control.

Additionally, in the context of stratified experiments, 
the way to aggregating evidence across all strata is largely unexplored.
Despite that the method in \cite{SL22quantile} applies to general stratified experiments with strata of varying sizes, their numerical experiments focused mainly on strata of equal sizes, where it is natural and reasonable to weight all strata equally.
However, in practice,  
the strata in a stratified experiment often have varying sizes, under which weighting becomes a critical issue.
This issue can become more severe when we consider general rank statistics within each stratum. 

Throughout the paper, we will investigate both of the above issues regarding the
choice of rank statistics and weighting in stratified experiments. We will
first focus on the CRE, where we propose to adaptively choose the rank
statistics based on the observed data and the quantile of interest. We will then
extend the approach to the SRE, incorporating weighting schemes to combine rank
statistics from strata of varying sizes.

\section{Adaptive Randomization Tests in Completely Randomized Experiments using Multiple Rank Statistics}\label{sec:CRE}

In this section, we focus on the CRE, where $n_\treat$ units are randomly assigned to the treatment group and the remaining $n_\control$ units are assigned to the control group. 
Below we first review the existing approaches from \citet{CDLM21quantile} and \citet{ZL24quantile}, and then study how to adaptively choose or combine multiple rank statistics. 

\subsection{Randomization test for effect quantiles with a general distribution-free statistic}\label{sec:single_test}

Consider a general distribution-free statistic $t(\bs{z}, \bs{y})$ under the CRE, where the distribution of $t(\bs{Z}, \bs{y})$ does not depend on the outcome vector $\bs{y}$. Define a tail-probability function, $G^{(t)}(c) \equiv \Pr\{t(\bs{Z}, \bs{y})\ge c\}$, where $\bs{y}$ can be any constant vector in $\mathbb{R}^n$ due to the distribution-free property and $G^{(t)}$ uses the superscript $t$ to refer to the particular test statistic used. 
For any $0\le k \le n_\treat$ and $c\in \mathbb{R}$,
\citet{ZL24quantile} extended \citet{CDLM21quantile} and proposed the following $p$-value for testing the null hypothesis $H_{k, c}^\treat$ in \eqref{eq:H_kc_t} based on a distribution-free test statistic $t(\bs{z}, \bs{y})$: 
\begin{align}\label{eq:pval_t}
    p^{\treat(t)}_{k, c} & \equiv G^{(t)} \Big( \inf_{\bs{\delta} \in \mathcal{H}_{k,c}^\treat} t(\bs{Z}, \bs{Y} - \bs{Z}\circ \bs{\delta} ) \Big),
\end{align}
recalling that $\mathcal{H}_{k,c}^\treat$ is the set of possible values of individual treatment effect vector $\bs{\tau}$ under  $H_{k,c}^\treat$.
In \eqref{eq:pval_t}, $\bs{Y} - \bs{Z}\circ \bs{\delta}$ represents the imputed control potential outcome vector when the individual treatment effects are hypothesized to be $\bs{\delta}$. When the hypothesized vector of treatment effects is the same as the truth, $\bs{Y} - \bs{Z}\circ \bs{\delta} = \bs{Y} - \bs{Z}\circ \bs{\tau} =\bs{Y}(0)$.
Thus, if the null hypothesis $H_{k,c}^\treat$ is true, then we must have 
$\inf_{\bs{\delta} \in \mathcal{H}_{k,c}^\treat} t(\bs{Z}, \bs{Y} - \bs{Z}\circ \bs{\delta} ) \le t(\bs{Z}, \bs{Y} - \bs{Z} \circ \bs{\tau}) = t(\bs{Z}, \bs{Y}(0))$. 
Note that $G^{(t)}(\cdot)$ is also the tail probability of $t(\bs{Z}, \bs{Y}(0))$ and is a non-increasing function. 
Therefore, $p^{\treat(t)}_{k, c} \ge G^{(t)}(t(\bs{Z}, \bs{Y}(0)))$ is stochastically larger than or equal to $\text{Uniform}(0,1)$, the uniform distribution on $(0,1)$. These imply the validity of the $p$-value in \eqref{eq:pval_t}.

\citet{CDLM21quantile} considered the following class of rank sum statistics, which are distribution-free under the CRE with appropriate handling of ties discussed shortly:
\begin{align}\label{eq:rank_sum}
    t(\bs{z}, \bs{y}) = \sum_{i=1}^n z_i \phi( \rank_i(\bs{y}) ),
\end{align}
where $\bs{z}\in \{0,1\}^n$ is a general treatment assignment vector, $\bs{y}\in \mathbb{R}^n$ is a general outcome vector,
$\rank_i(\bs{y})$ is the rank of the $i$th coordinate of $\bs{y}$ among all coordinates of $\bs{y}$,
and $\phi$ is a general non-decreasing rank transformation. 
Moreover, we break ties randomly, or based on the indices assuming that the units have been randomly shuffled before analysis. 
Specifically, 
$\rank_i(\bs{y}) < \rank_j(\bs{y})$ if (i) $y_i < y_j$ or (ii) $y_i=y_j$ and $i<j$. 
This is critical for ensuring the distribution-free property of the rank sum statistic in \eqref{eq:rank_sum} under the CRE. 
When $\phi$ is the identity function, \eqref{eq:rank_sum} becomes the usual Wilcoxon rank sum statistic. 
When $\phi$ is
\begin{align}\label{eq:stephenson}
    \phi(r) = \binom{r-1}{\zeta-1} \text{ for } r\ge \zeta, \text{ and }  \phi(r) = 0 \text{ otherwise},
\end{align}
\eqref{eq:rank_sum} becomes the Stephenson rank sum statistic, where $\zeta$ is
a parameter that needs to be prespecified.\footnote{We discuss other
choices of rank transformations that are more appropriate for stratified experiments in Section 
\ref{sec:combine_sre}.} When $\zeta=2$, the Stephenson rank
is equivalent to the Wilcoxon rank under the CRE. 
When $\zeta>2$, the Stephenson rank amplifies the
contribution of high-ranking outcomes. The parameter $\zeta$ controls how
much we care about the extremes: larger $\zeta$ means we weight the very top
ranks far more heavily than the rest.
For example, 
when $\zeta = 30$, only units ranking 30th or
higher contribute, and the highest ranks contribute far more than ranks
just above the threshold. 
This emphasis on extremes can be beneficial for our tests of effect quantiles, where we deliberately shift some treated units with the largest outcomes to the smallest, as discussed below; in particular, the test can still be significant if there remain some treated units with outcomes larger than most controls.

For the rank sum statistic in \eqref{eq:rank_sum}, 
\citet{CDLM21quantile} and \citet{ZL24quantile} showed that  the $p$-value in \eqref{eq:pval_t} can be efficiently computed. 
Specifically, 
the statistic $t(\bs{Z}, \bs{Y} - \bs{Z}\circ \bs{\delta} )$ in \eqref{eq:pval_t} is minimized at $\bs{\delta} = \bs{\xi}_{k,c}^\treat = (\xi_{k,c,1}^\treat, \xi_{k,c,2}^\treat, \ldots, \xi_{k,c,n}^\treat)^\top$, where
\begin{align}\label{eq:xi_kc_t}
    \xi_{k,c,i}^\treat = 
    \begin{cases}
        \infty, & \text{if } Z_i = 1 \text{ and } 
        \sum_{j=1}^n Z_j \I \{ \rank_j(\bs{Y})  \ge \rank_i(\bs{Y}) \} \le n_{\treat} - k, \\
        c, & \text{otherwise}.
    \end{cases}
\end{align}
Note that the imputed control potential outcome $\bs{Y}-\bs{Z}\circ \bs{\delta}$ depends only on hypothesized effects (i.e., coordinates of $\bs{\delta}$) for treated units.
Intuitively, to achieve the infimum,
we let the $n_\treat - k$ treated units with the largest observed outcomes have infinitely large individual effects,
and the remaining treated units have individual effects $c$.
In this way, the imputed control potential outcomes for treated units are made as small as possible, thereby minimizing the rank sum statistic.

\subsection{Minimum $p$-value approach}\label{subsec:min_pval}

As demonstrated by simulation in \citet{CDLM21quantile} and in Section \ref{subsec:cre_sim} of the supplementary material, the power of the $p$-value in \eqref{eq:pval_t} for testing the null hypothesis on treatment effect quantiles 
depends crucially on the choice of the rank sum statistic $t(\cdot, \cdot)$ in \eqref{eq:rank_sum} or equivalently the rank transformation $\phi(\cdot)$.
Below we study how to adaptively choose the rank transformation 
based on the observed data.

Suppose we have $H$ rank sum statistics of form \eqref{eq:rank_sum} under consideration. 
Denote these $H$ rank sum statistics by $t^{(1)}(\cdot, \cdot), \ldots, t^{(H)}(\cdot, \cdot)$, and the corresponding monotone nondecreasing rank transformations by $\phi^{(1)}(\cdot), \ldots, \phi^{(H)}(\cdot)$. That is,
$t^{(h)}(\bs{z}, \bs{y}) = \sum_{i=1}^n z_i \phi^{(h)}(\rank_i(\bs{y}))$ for $1\le h \le H$.
For each $h$, 
let $G^{(h)}(c) = \Pr\{ t^{(h)}(\bs{Z}, \bs{y}) \ge c \}$ 
denote the corresponding tail probability function,  
where $\bs{y}$ can be any constant vector in $\mathbb{R}^n$ due to the distribution-free property. 
For any $0\le k \le n_\treat$,  $c\in \mathbb{R}$ and $1\le h\le H$, 
let $p^{\treat(h)}_{k, c} \equiv G^{(h)}( \inf_{\bs{\delta} \in \mathcal{H}_{k,c}^\treat} t^{(h)}(\bs{Z}, \bs{Y} - \bs{Z}\circ \bs{\delta} ))$ 
denote the $p$-value for testing 
the null hypothesis $H_{k, c}^\treat$ in \eqref{eq:H_kc_t}, defined analogously as in \eqref{eq:pval_t}.

From the discussion in Section \ref{sec:single_test}, each single $p^{\treat(h)}_{k, c}$ is a valid $p$-value for testing the null hypothesis $H_{k,c}^\treat$ in \eqref{eq:H_kc_t}. 
However, when the null hypothesis is false, different test statistics can have different power. 
Ideally, we want to use the minimum $p$-value $p^{\treat,\min}_{k,c}\equiv \min_{1\le h \le H} p^{\treat(h)}_{k,c}$, which would have the greatest power against the null hypotheses. 
However, in general, such a minimum $p$-value will inflate the type-I error rate and is thus invalid. 
The following theorem characterizes the distribution of the minimum $p$-value $p^{\treat,\min}_{k,c}$,  which allows us to control the type-I error rate.

\begin{theorem}\label{thm:cre_min_pval}
Under the CRE,
for any $0\le k\le n_\treat$, $c\in \mathbb{R}$ and $0\le \alpha \le 1$, 
$\Pr( p^{\treat, \min}_{k,c} \le \alpha \text{ and } H_{k,c}^\treat \text{ holds}) \le \overline{F}(\alpha)$. 
Here $\overline{F}(\cdot)$ is
the distribution function
of $\overline{t}(\bs{Z}, \bs{y})$, where $\overline{t}(\bs{Z}, \bs{y}) \equiv \min_{1\le h \le H} G^{(h)}(t^{(h)}(\bs{Z}, \bs{y}))$ is distribution-free, in the sense that 
$\overline{F}(\alpha) = \Pr\{ \overline{t}(\bs{Z}, \bs{y}) \le \alpha \}$ 
does not depend on the value of  $\bs{y} \in \mathbb{R}^n$.
\end{theorem}

Theorem \ref{thm:cre_min_pval} provides an upper bound on the type-I error rate when
rejecting $H_{k,c}^\treat$ using the minimum $p$-value such that
$p^{\treat,\min}_{k,c} \leq \alpha$. Based on this, we can adjust the cutoff
$\alpha$ such that $\overline{F}(\alpha)$ is bounded by the desired
significance level. 
Alternatively, we can  apply $\overline{F}$ directly to the observed minimum
$p$-value to obtain a valid $p$-value, as stated in the following corollary.

\begin{corollary}\label{cor:cre_min_pval}
Under the CRE, for any $0\le k\le n_\treat$ and $c\in \mathbb{R}$, $\overline{F}( p^{\treat,\min}_{k,c})$ is a valid $p$-value for testing the null hypothesis $H_{k,c}^\treat$ in \eqref{eq:H_kc_t},
    where $\overline{F}(\cdot)$ is defined as in Theorem \ref{thm:cre_min_pval}.
\end{corollary}

The calibration absorbs the multiplicity, and is sharper than the naive Bonferroni correction.

\subsection{Combining multiple rank sum  statistics}\label{subsec:combine_test}

In practice, another popular way to combining multiple tests is to construct a new test statistic from the multiple test statistics under consideration. 
Indeed, 
for testing the null hypothesis $H_{k,c}^\treat$ in \eqref{eq:H_kc_t}, 
the valid $p$-value from Corollary \ref{cor:cre_min_pval} can be equivalently written as a $p$-value of form \eqref{eq:pval_t} with a combined distribution-free rank statistic defined in the theorem below.

\begin{theorem}\label{thm:cre_min_pval_combined}
Under the CRE, 
for any $1\le k\le n_\treat$ and $c\in \mathbb{R}$, 
the valid $p$-value $\overline{F}(p^{\treat, \min}_{k,c})$ in Corollary \ref{cor:cre_min_pval} for testing the null hypothesis $H_{k,c}^\treat$ in \eqref{eq:H_kc_t} is equivalent to the $p$-value $p_{k,c}^{\treat(t)}$ in \eqref{eq:pval_t} with 
$t(\bs{z}, \bs{y}) = - \overline{t}(\bs{z}, \bs{y})$, 
where $\overline{t}(\cdot, \cdot)$ is defined as in Theorem \ref{thm:cre_min_pval} and is distribution-free under the CRE. 
Specifically, 
\begin{align}\label{eq:pval_combined_minpval}
    \overline{F}(p^{\treat, \min}_{k,c})
    = p_{k,c}^{\treat(-\overline{t})} \equiv 
    G^{(-\overline{t})} \Big( \inf_{\bs{\delta} \in \mathcal{H}_{k,c}^\treat} -\overline{t}(\bs{Z}, \bs{Y} - \bs{Z}\circ \bs{\delta} ) \Big)
    = 
    G^{(-\overline{t})} \big( -\overline{t}(\bs{Z}, \bs{Y} - \bs{Z} \circ \bs{\xi}_{k,c}^\treat ) \big),
\end{align}
where $G^{(-\overline{t})}(c) \equiv \Pr\{ -\overline{t}(\bs{Z}, \bs{y} ) \ge c \} $ 
with any constant $\bs{y} \in \mathbb{R}^n$, 
and $\bs{\xi}_{k,c}^\treat$ is defined in \eqref{eq:xi_kc_t}.
\end{theorem}

Theorem~\ref{thm:cre_min_pval_combined} shows that the calibrated minimum $p$-value is simply the standard $p$-value for a new combined test statistic, $-\overline{t}$. Thus, rather than running multiple tests and correcting, we can perform a single test whose statistic captures whichever pattern---Wilcoxon-like or Stephenson-like---is strongest in the data. The combined statistic remains distribution-free under the CRE, so validity follows directly.

\begin{remark}
An important reason for the equivalence in Theorem \ref{thm:cre_min_pval_combined} is that the optimization in the $p$-value $p^{\treat(h)}_{k, c}$ defined as in \eqref{eq:pval_t} for each individual rank sum statistic $t^{(h)}(\cdot, \cdot)$ is attained at the same hypothesized individual treatment effect vector $\bs{\xi}_{k,c}^\treat$.
Note that this may no longer be true for stratified experiments, an issue we will study in Section \ref{sec:SRE}.
\end{remark}

Below we give more intuition for the combined rank statistic $-\overline{t}(\bs{z}, \bs{y})$. 
By definition, it has the following equivalent forms, as well as approximations when the sample size is large: 
\begin{align}\label{eq:combined_stat_approx}
  - \overline{t}(\bs{z}, \bs{y})
  & = - \min_{1\le h \le H} G^{(h)}(t^{(h)}(\bs{z}, \bs{y}))
  = 
  \max_{1\le h\le H} 
  - G^{(h)}(t^{(h)}(\bs{z}, \bs{y}))
  \nonumber
  \\
  & \approx 
  \max_{1\le h\le H} 
  - \overline{\Phi}\left(\frac{t^{(h)}(\bs{z}, \bs{y})-\mu^{(h)}}{\sigma^{(h)}}\right)
  = 
  - \overline{\Phi}\left(
  \max_{1\le h\le H} \frac{t^{(h)}(\bs{z}, \bs{y})-\mu^{(h)}}{\sigma^{(h)}} \right),
\end{align}
where $\overline{\Phi}(\cdot)$ is a tail probability function of the standard Gaussian distribution,
and $\mu^{(h)}$ and $\sigma^{(h)}$ are the mean and standard deviation of $t^{(h)}(\bs{Z}, \bs{y})$ under the CRE, for all $1\le h\le H$.
In \eqref{eq:combined_stat_approx},
we use the Gaussian approximation for the tail probability function $G^{(h)}(\cdot)$, whose validity can be justified by the finite population central limit theorem \citep{hajek1960limiting, LD2017, shi2024some}.
Note that the $p$-value in \eqref{eq:pval_t} is invariant to monotone strictly increasing transformations of the test statistic. 
Therefore, the valid $p$-value in Corollary \ref{cor:cre_min_pval} and Theorem \ref{thm:cre_min_pval_combined} is 
approximately 
equivalent to the $p$-value in \eqref{eq:pval_t} using the following maximum of the standardized rank sum statistics, which is also distribution-free under the CRE:
\begin{align}\label{eq:combined_stat_max}
    \tilde{t}(\bs{z}, \bs{y}) = \max_{1\le h\le H} \frac{t^{(h)}(\bs{z}, \bs{y})-\mu^{(h)}}{\sigma^{(h)}}.
\end{align}
Furthermore, the corresponding $p$-value in \eqref{eq:pval_t}
can also be efficiently computed.

\begin{theorem}\label{thm:cre_min_pval_combined_gaussian}
    Under the CRE, 
    for any $0\le k\le n_\treat$ and $c\in \mathbb{R}$, 
    the $p$-value $p_{k,c}^{\treat(\tilde{t})}$ defined as in \eqref{eq:pval_t} with $\tilde{t}(\cdot, \cdot)$ in \eqref{eq:combined_stat_max} is valid for testing the null $H_{k,c}^\treat$ in \eqref{eq:H_kc_t}, and it has the following equivalent forms:    \begin{align}\label{eq:pval_combined_minpval_gaussian}
        p_{k,c}^{\treat(\tilde{t})} \equiv G^{(\tilde{t})} \Big( \inf_{\bs{\delta} \in \mathcal{H}_{k,c}^\treat} \tilde{t}(\bs{Z}, \bs{Y} - \bs{Z}\circ \bs{\delta} ) \Big)
        = 
        G^{(\tilde{t})} \big( \tilde{t}(\bs{Z}, \bs{Y} - \bs{Z}\circ \bs{\xi}_{k,c}^\treat ) \big),
    \end{align}
    where $G^{(\tilde{t})}(c) \equiv \Pr\{\tilde{t}(\bs{Z}, \bs{y} ) \ge c \} $ 
    with $\bs{y}$ being any constant in $\mathbb{R}^n$, 
    and $\bs{\xi}_{k,c}^\treat$ is defined in \eqref{eq:xi_kc_t}.
\end{theorem}

We want to emphasize that the Gaussian approximation in
\eqref{eq:combined_stat_approx} does not affect the finite-sample validity of
the $p$-value in Theorem \ref{thm:cre_min_pval_combined_gaussian}. We use it as
a heuristic to motivate another way of combining multiple rank statistics.
Furthermore, as discussed later in Section \ref{sec:SRE}, such a Gaussian
approximation for combining multiple rank statistics can be computationally
useful for stratified experiments. 

\subsection{Simultaneous prediction and confidence intervals for treatment effect quantiles}

Following \citet{ZL24quantile}, we can invert the test for the null hypothesis in \eqref{eq:H_kc_t} over all $c\in \mathbb{R}^n$ to construct prediction intervals for treatment effect quantiles among treated units. 
These prediction intervals are simultaneously valid across all effect quantiles among treated units. 
Furthermore, we can construct simultaneous prediction intervals for treatment effect quantiles among control units, and combine prediction intervals for treated and control units to obtain simultaneous confidence intervals for treatment effect quantiles among all units.
These simultaneous prediction and confidence intervals for quantiles of individual effects among treated, control, or all units essentially provide simultaneous confidence bands for the entire quantile or distribution functions of individual effects for these groups. 
We summarize the results below for completeness; they are direct corollaries from \citet{ZL24quantile}.

\begin{corollary}\label{cor:interval_cre}
Under the CRE, consider any distribution-free statistic $t(\cdot, \cdot)$ and the $p$-value $p^{\treat(t)}_{k, c}$ in \eqref{eq:pval_t}. For example, $t(\cdot, \cdot)$ can be either $-\overline{t}(\cdot, \cdot)$ or $\tilde{t}(\cdot, \cdot)$ in Theorems \ref{thm:cre_min_pval_combined} or \ref{thm:cre_min_pval_combined_gaussian}, with the corresponding $p$-value $p^{\treat(t)}_{k, c}$ being that shown in \eqref{eq:pval_combined_minpval} or \eqref{eq:pval_combined_minpval_gaussian}. 
    \begin{enumerate}[label=(\roman*)]
        \item For any $1\le k \le n_\treat$ and $\alpha \in (0,1)$, $\mathcal{I}_{\treat(k)}^\alpha := \{c : p^{\treat(t)}_{k, c} > \alpha, c \in \mathbb{R}\}$ is a $1 - \alpha$ prediction set for $\tau_{\treat(k)}$, and it has the form $(\underline{c}, \infty)$ or $[\underline{c}, \infty]$, with $\underline{c} = \inf\{c: p^{\treat(t)}_{k, c} > \alpha \}$. 
        Moreover, these prediction intervals are simultaneously valid across all $1\le k \le n_\treat$, i.e., 
        $\Pr(\tau_{\treat(k)}\in \mathcal{I}_{\treat(k)}^\alpha \text{ for all } 1\le k \le n_\treat)\ge 1-\alpha$.

        \item By switching the treatment labels and changing the outcome signs, we can analogously construct simultaneous prediction intervals $\mathcal{I}_{\control(k)}^\alpha$s for $\tau_{\control(k)}$s among control units.

        \item 
        Pool the prediction intervals $\mathcal{I}_{\treat(k)}^\alpha$s and $\mathcal{I}_{\control(k)}^\alpha$s for effect quantiles among treated and control units,
        and sort them so that they are nested, i.e., $\mathcal{I}_{(1)}^\alpha \supset \mathcal{I}_{(2)}^\alpha \supset \ldots \supset \mathcal{I}_{(n)}^\alpha$. 
        Then $\mathcal{I}_{(k)}^\alpha$s are $1-2\alpha$ simultaneous confidence intervals for $\tau_{(k)}$s among all units, i.e., $\Pr\{ \tau_{(k)}\in \mathcal{I}_{(k)}^\alpha \text{ for all } 1\le k \le n \} \ge 1-2\alpha$. 

    \end{enumerate}
\end{corollary}

In other words, by inverting the test across all values of $c$ for each $k$, we
obtain prediction intervals for every quantile of the treatment effect
distribution among treated units---and these intervals are simultaneously valid
across all quantiles. The same procedure, with relabeled treatments and changed outcome sign, gives
simultaneous prediction intervals for control units. Pooling and sorting the
intervals for both groups yields simultaneous confidence intervals for the quantiles among
all $n$ units. 
We illustrate these intervals using the teacher training experiment in Section \ref{sec:app}.

\section{Weighting and Combination of Stratified Rank Statistics in Stratified Randomized Experiments}\label{sec:SRE}

We now consider stratified experiments, dividing the $n$ units into $S$ strata, with $n_s$ units in each stratum $s$ for $1\le s\le S$. 
For descriptive convenience, 
we will index each unit by $si$, which denotes the $i$th unit in the $s$th stratum, 
and use analogously $Y_{si}(1)$, $Y_{si}(0)$, $\tau_{si}$, $Z_{si}$ and $Y_{si}$ to denote its potential outcomes, individual treatment effect, treatment assignment and observed outcome, respectively. We also introduce $\bs{Z}_s = (Z_{s1}, Z_{s2} \ldots , Z_{sn_s})$ and $\bs{Y}_s = (Y_{s1}, Y_{s2} \ldots , Y_{sn_s})$ to denote the vectors of treatment assignments and observed outcomes for units within stratum $s$. 
Under the SRE, the treatment assignments are mutually independent across strata, and, within each stratum $s$, the treatment assignments are from a CRE with $n_{s\treat}$ treated units and $n_{s0} \equiv n_{s} - n_{s\treat}$ control units, where $n_{s\treat}$ and $n_{s0}$ are two fixed positive integers. 

In the following, we first review the existing approaches from \citet{SL22quantile} and \citet{ZL24quantile}. We then study methods for weighting rank statistics across strata. Finally, we investigate how to combine multiple rank statistics.

\subsection{Randomization test with a general rank statistic}\label{sec:single_rank_stat_stratified}

Consider a general distribution-free statistic under the SRE such that the distribution of $t(\bs{Z}, \bs{y})$ is the same for all constant $\bs{y} \in \mathbb{R}^n$.
For any $1\le k \le n_\treat$ and $c\in \mathbb{R}$, 
\citet{ZL24quantile} proposed a valid $p$-value of the same form as in \eqref{eq:pval_t} for testing the null hypothesis $H_{k,c}^\treat$ in \eqref{eq:H_kc_t}.
The key difference from the discussion in Section \ref{sec:CRE} about the CRE is that test statistic $t(\cdot, \cdot)$ needs to be distribution-free under the SRE and its tail probability function $G^{(t)}(\cdot)$ is defined in terms of the random treatment assignment vector from the SRE.

\cite{SL22quantile} considered the following class of stratified rank sum statistics in which we first rank outcomes within each stratum and then sums the contributions across all strata:
\begin{align}\label{eq:rank_sum_stratified}
    t(\bs{z}, \bs{y}) = \sum_{s = 1}^S t_s(\bs{z}_s, \bs{y}_s) = \sum_{s=1}^S \sum_{i=1}^{n_s} z_{si} \phi_s(\rank_i(\bs{y}_s)),
\end{align}
where each $t_s(\cdot, \cdot)$ is a rank sum statistic as in \eqref{eq:rank_sum} for the CRE within stratum $s$.
In \eqref{eq:rank_sum_stratified}, $\phi_s(\cdot)$s are monotone nondecreasing rank transformations, 
and ties in ranking are broken 
based on units' indices after random shuffling within each stratum.

Compared to the CRE, the optimization for the valid $p$-value in \eqref{eq:pval_t} becomes more challenging under the SRE. 
In particular, for the stratified rank sum statistic in \eqref{eq:rank_sum_stratified},
there is no simple closed-form solution for minimizing $t(\bs{Z}, \bs{Y}-\bs{Z}\circ\bs{\delta})$ subject to $\bs{\delta} \in \mathcal{H}_{k,c}^\treat$. 
\cite{SL22quantile} transformed this optimization problem into instances of the multiple-choice knapsack problem, which can be exactly solved in $O(n^2)$ time and conservatively in $O(n)$ time, where the latter still leads to valid but more conservative $p$-values. 
We provide some intuition for their algorithm here. 
The constraint $\bs{\delta} \in \mathcal{H}_{k,c}^\treat$ requires that at most $n_\treat - k$ treated units across all strata have individual
effects greater than $c$, with the remaining treated units having effects at most $c$. Recall that hypothesized effects for control units do
not affect $t(\bs{Z}, \bs{Y}-\bs{Z}\circ\bs{\delta})$.
The question then becomes how to allocate these $n_\treat - k$ ``large-effect'' slots across strata in order to minimize the test statistic.
This allocation can be characterized by  
$(k_1, k_2 \ldots, k_S)$ such that 
(i) $0 \le k_s \le n_{s\treat}$ for all $s$, 
(ii) $\sum_{s=1}^S k_s = k$, 
and (iii) there are at most $n_{s\treat}-k_s$ treated units in stratum $s$ with individual effects greater than $c$, for all $1\le s\le S$. 
Once we know $(k_1, k_2 \ldots, k_S)$, from the discussion on rank sum statistics in Section \ref{sec:single_test}, the infimum of $t(\bs{Z}, \bs{Y}-\bs{Z}\circ\bs{\delta})$ is achieved when, for all $s$, the hypothesized effects $\bs{\delta}_s$ for units in stratum $s$ are equal to $\bs{\xi}_{s, k,c}^\treat$, which is defined analogously as in \eqref{eq:xi_kc_t} but only for units within stratum $s$.
Therefore, 
the optimization problem essentially simplifies to finding the worst-case configuration of $(k_1, k_2 \ldots, k_S)$. 
This turns out to be a multiple-choice knapsack problem, which can be solved exactly by integer linear programming or dynamic programming, 
and conservatively by linear programming or a greedy algorithm; 
see \cite{SL22quantile} for details.

\subsection{Normalized ranks and weighting for varying strata sizes}\label{sec:weighting}

How should we choose a rank transformation when the goal is to combine results
across strata? When all strata have equal size, it is natural to use the same
rank transformation $\phi_s(\cdot)$ for each stratum and aggregate the rank
statistics with equal weights, as in \citet{SL22quantile}. However, the
situation becomes more complicated when strata sizes differ. For
example, consider using the Stephenson rank transformation for all strata with the same
parameter $\zeta$ as in \eqref{eq:stephenson}. First, strata with sizes less
than $\zeta$ will always have zero rank sum, and are thus dropped from 
the analysis. Second, rank sum statistics across strata can have different scales; for example, when $\zeta=5$,
the transformed rank of a unit can be as large as $3876$ in a stratum of size
$20$, whereas it is at most $126$ in a stratum of size $10$. Therefore,
without appropriate weighting, the stratified rank sum statistic in
\eqref{eq:rank_sum_stratified} may discard smaller strata and thus lose power.

To address this issue, we draw upon the literature on the classical nonparametric two-sample tests using rank sum statistics, particularly \citet{van1960combination} and \citet{puri1965combination}. 
The literature considers mainly superpopulation settings with i.i.d.~samples within each stratum and treatment group, and focuses on testing whether the distributions of the two treatment groups are the same within each stratum. 
Despite the difference from our setting with a finite population of experimental units and random treatment assignment, 
the results from that literature provides useful insights for our context. 
In particular, the testing procedure for the equality of distributions within each stratum is equivalent to our randomization test for whether the maximum individual treatment effect among treated units is bounded by zero, i.e., the null hypothesis in \eqref{eq:H_kc_t} with $k=n_\treat$ and $c=0$.
Following \citet{puri1965combination}, we introduce the following normalized rank:
\begin{align}\label{eq:normalized_rank}
    \widetilde{\rank}_i(\bs{y}) = \frac{\rank_i(\bs{y})}{m+1},
    \quad
    (\text{for integer } m\ge 1, \bs{y}\in \mathbb{R}^m, \text{ and } 1\le i\le m ),
\end{align}
which always takes value between $0$ and $1$,
and rewrite the general stratified rank sum statistic in \eqref{eq:rank_sum_stratified} in the following form:
\begin{align}\label{eq:rank_sum_stratified_normalized}
    t(\bs{z}, \bs{y}) = \sum_{s=1}^S w_s t_s(\bs{z}_s, \bs{y}_s),
    \text{ with }
    t_s(\bs{z}_s, \bs{y}_s) =
    \frac{1}{n_{s\treat}} \sum_{i=1}^{n_s} z_{si} \phi_s\big( \widetilde{\rank}_i(\bs{y}_s) \big) \text{ for all } s,
\end{align}
where $\phi_s(\cdot)$s are monotone nondecreasing functions from $[0,1]$ to $\mathbb{R}$ that transform normalized ranks. 
The statistic in \eqref{eq:rank_sum_stratified_normalized} also belongs to the class of stratified rank sum statistics in \eqref{eq:rank_sum_stratified}, by appropriately choosing the rank transformations for each stratum.
However, the form in \eqref{eq:rank_sum_stratified_normalized} has the 
advantage that the rank sum statistics $t_s(\cdot, \cdot)$s across strata are now on the same scale (assuming the transformations $\phi_s(\cdot)$s on $[0,1]$ share a common scale), and the weights $w_s$s can be chosen to balance the contribution from each stratum. For example, a unit ranked highest in a stratum of size 15 receives roughly the same normalized rank as one ranked highest in a stratum of size 43. 

We now discuss weights $w_s$s in \eqref{eq:rank_sum_stratified_normalized}. 
Consider first the case where the rank transformations $\phi_s(\cdot)$s are all identity functions, 
where \eqref{eq:rank_sum_stratified_normalized} reduces to a weighted stratified Wilcoxon rank sum statistic.
\citet{van1960combination} proposed the following two weighting schemes:
\begin{align}\label{eq:weighting_van_opt}
    \text{Scheme 1}: \quad w_s & =
    n_{s\treat}, \quad (s=1,2,\ldots,S),\\
    \label{eq:weighting_van_free}
    \text{Scheme 2}: \quad w_s & = \frac{n_s+1}{n_{s\control}}, \quad (s=1,2,\ldots,S).
\end{align}
Both schemes have theoretical justification, and they coincide when all strata are of equal size and have the same proportion of treated units. \citet{van1960combination} showed that Scheme 1 in \eqref{eq:weighting_van_opt} is asymptotically optimal under a class of local alternatives, and Scheme 2 in \eqref{eq:weighting_van_free} has a design-free property in the sense that the resulting test can be consistent regardless of the sizes of treated and control units within each stratum. These results apply to the usual superpopulation setting for stratified two-sample comparison, under either one of the following two scenarios for the asymptotics: 
\begin{enumerate}[topsep=1ex,itemsep=-0.3ex,partopsep=1ex,parsep=1ex
	]
    \item[(a)] the number of strata $S$ is fixed, and the size of each stratum goes to infinity;
    \item[(b)] the number of strata $S$ goes to infinity, and the size of each stratum remains bounded. 
\end{enumerate}
\citet{puri1965combination} extended \citet{van1960combination} to more general rank transformations. 
When all the rank transformations $\phi_s(\cdot)$ in \eqref{eq:rank_sum_stratified_normalized} equal to some common $\phi(\cdot)$, 
Scheme 1 in \eqref{eq:weighting_van_opt} is still asymptotically optimal under a class of local alternatives, and Scheme 2 in \eqref{eq:weighting_van_free} still has the design-free property by the same logic as \citet{van1960combination}. However, \citet{puri1965combination} established these results only under Scenario (a); whether they extend to Scenario (b) remains an open question beyond the scope of this paper.

The above results from \citet{van1960combination} and \citet{puri1965combination}, although derived under superpopulation settings with certain assumptions on the data-generating process, 
provide useful guidance for choosing weights $w_s$s in \eqref{eq:rank_sum_stratified_normalized} for the stratified rank sum statistics. 
We study the two weighting schemes in \eqref{eq:weighting_van_opt} and \eqref{eq:weighting_van_free}. 
We show finite-sample performance for inferring quantiles of treatment effects through simulation studies in the supplementary material.

\subsection{Combining multiple stratified rank sum statistics} \label{sec:combine_sre}

Recall that our proposal to increase power while avoiding multiple testing is to combine multiple rank statistics. How should these statistics be combined in stratified experiments?
We first discuss the choice of individual rank statistics, and then examine two approaches for combining them.

\subsubsection{Polynomial rank transformations}

From the discussion in Section \ref{sec:weighting}, we consider stratified rank sum statistics of form \eqref{eq:rank_sum_stratified_normalized} with the same rank transformation for all strata, i.e., $\phi_1 = \phi_2 = \ldots \phi_S = \phi$ for some $\phi(\cdot)$, 
and choose the weights $w_s$s based on \eqref{eq:weighting_van_opt} or \eqref{eq:weighting_van_free}. 
To make the test more sensitive to the tails of the outcomes as the Stephenson rank test,
we consider the following polynomial transformation for the normalized ranks:
\begin{align}\label{eq:polynomial}
    \phi(x) = x^{\zeta-1}, \quad (\text{for some } \zeta \ge 1).
\end{align}
In other words, raising the normalized rank to a power amplifies the contrast between high and low ranks, analogous to the Stephenson rank transformation in \eqref{eq:stephenson} while avoiding zero ranks for small strata. When $\zeta = 2$, the transformation is linear and all ranks contribute proportionally. When $\zeta = 10$, a unit at the 90th percentile contributes about 39\% as much as one at the 100th percentile, while one at the 50th percentile contributes only 0.2\%. The polynomial transformation with normalized ranks thus achieves the tail-sensitivity of the Stephenson rank transformation while remaining smooth across strata of different sizes.

Similar to Sections \ref{subsec:min_pval} and \ref{subsec:combine_test}, 
we are interested in adaptively selecting from multiple stratified rank sum statistics, as their power can vary under different true data-generating processes. 
Specifically, let $t^{(1)}(\cdot, \cdot), \ldots, t^{(H)}(\cdot, \cdot)$ be $H$ stratified rank sum statistics under consideration, 
each of which has the form in \eqref{eq:rank_sum_stratified_normalized} with some rank transformation $\phi_s^{(h)}$s and weights $w_s^{(h)}$s, for $1\le h\le H$. 
From Section \ref{sec:weighting}, we will consider mainly the case where each stratified rank sum statistic uses a common rank transformation $\phi^{(h)}$ for all strata with   weights $w_s^{(h)}$ based on \eqref{eq:weighting_van_opt} or \eqref{eq:weighting_van_free}. 
Nevertheless, here we allow arbitrary stratum-specific rank transformations, as they do not complicate the implementation of the resulting randomization test using a combination of these test statistics. 

In the following two subsections, we consider two ways to combining these $H$ rank statistics, respectively. In the first approach, we directly combine the $H$ stratified rank sum statistics. In the second approach, we consider first combining the $H$ rank sum statistics within stratum, and then aggregate them across all the strata. 
In other words, 
in the first approach, we first aggregate across all the strata using a single rank transformation and then combine multiple rank transformations; whereas in the second approach, we combine multiple rank transformations within each stratum and then aggregate them across strata. 
We show the performance of these two approaches in the empirical application in Section \ref{sec:sredata}, and in simulation studies reported in the supplementary material: in both the application and simulations, the second approach performs better than the first.

\subsubsection{Combination of the aggregated statistics from all strata} \label{sec:comb1}

When considering the CRE, we discussed two approaches for combining multiple rank sum statistics: one uses the minimum $p$-value with careful calibration, and the other uses the maximum of these rank sum statistics after proper standardization. 
These two approaches can be equivalent under the CRE. 
This equivalence, however, may not hold under the SRE. 
The reason is that, under the SRE, the $p$-value in \eqref{eq:pval_t} for each stratified rank sum statistic may be achieved at different values of the hypothesized effect $\bs{\delta}$, or more specifically, at different configurations of $(k_1, \ldots, k_S)$ in the corresponding multiple-choice knapsack problem as discussed in Section \ref{sec:single_rank_stat_stratified}.
Consequently, under the SRE, using the minimum $p$-value approach as in Theorem \ref{thm:cre_min_pval} can be more conservative than using the combined test statistic approach as in Theorem \ref{thm:cre_min_pval_combined}. 

Similar to Theorem \ref{thm:cre_min_pval_combined}, we consider the following combined stratified rank sum statistic:
\begin{align}\label{eq:combined_statistic_tail_prob_std}
    -\overline{t}(\bs{z}, \bs{y}) = - \min_{1\le h \le H} G^{(h)}(t^{(h)}(\bs{z}, \bs{y}))
    =
    \max_{1\le h \le H} \big\{ - G^{(h)}(t^{(h)}(\bs{z}, \bs{y})) \big\},
\end{align}
where $G^{(h)}(\cdot)$ is the tail probability function for each individual stratified rank sum statistic $t^{(h)}(\bs{Z}, \bs{y})$ under the SRE.

The statistic in \eqref{eq:combined_statistic_tail_prob_std} is still distribution-free under the SRE, and the $p$-value in \eqref{eq:pval_t} with $t=-\overline{t}$ is still valid for testing the null hypothesis $H_{k, c}^\treat$ in \eqref{eq:H_kc_t}.
However, calculating this $p$-value---equivalently, minimizing $-\overline{t}(\bs{Z}, \bs{Y} - \bs{Z}\circ \bs{\delta})$ over $\bs{\delta} \in \mathcal{H}_{k,c}^\treat$---is computationally demanding. 
Even a single stratified rank statistic requires solving a knapsack problem, and the combined statistic further complicates this task: different component statistics may attain their minima at different stratum allocations $(k_1, \ldots, k_S).$

To simplify the computation, we first invoke the Gaussian approximation for the combined stratified rank sum statistic in \eqref{eq:combined_statistic_tail_prob_std}. 
Under the SRE, for any constant $\bs{y} \in \mathbb{R}^n$, each stratified rank sum statistic $t^{(h)}(\bs{Z}, \bs{y})$ has mean and variance 
\begin{align*}
    \mu^{(h)} = \sum_{s=1}^S w_s^{(h)} \mu_{s}^{(h)}, 
    \quad 
    \big(\sigma^{(h)}\big)^2 & =\sum_{s=1}^S (w_s^{(h)})^2 \big(\sigma_{s}^{(h)}\big)^2, 
    \qquad (1\le h \le H),
\end{align*}
where, for $1\le s\le S$ and $1\le h \le H$, 
\begin{align}\label{eq:mu_s}
    \mu_{s}^{(h)} & = \E \{t^{(h)}_s(\bs{Z}, \bs{y})\} = \frac{1}{n_s} \sum_{i=1}^{n_s} \phi_s^{(h)}\left( \frac{i}{n_s+1} \right), 
    \\
    \label{eq:sigma_s}
    \big(\sigma_{s}^{(h)}\big)^2 & = \Var\{ t_s^{(h)}(\bs{Z}, \bs{y}) \}
    = 
    \frac{n_{s\control}}{n_{s\treat}n_s(n_s-1)} \sum_{i=1}^{n_s} \left\{ \phi_s^{(h)}\left( \frac{i}{n_s+1} \right) - \mu_s^{(h)} \right\}^2, 
\end{align}
We then consider the following maximum of the standardized stratified rank sum statistics:
\begin{align}\label{eq:combined_stat_max_sre}
    \tilde{t}(\bs{z}, \bs{y}) = \max_{1\le h\le H} \frac{t^{(h)}(\bs{z}, \bs{y})-\mu^{(h)}}{\sigma^{(h)}},
\end{align}
which is again distribution-free under the SRE.

By the same logic as in \eqref{eq:combined_stat_approx} and the finite population central limit theorem for the SRE \citep{Bickel1984, LY2020stratifiedadj}, 
$
-\overline{t}(\bs{z}, \bs{y}) \approx 
-\overline{\Phi}(\tilde{t}(\bs{z}, \bs{y}))$. 
Thus, the randomization $p$-values in \eqref{eq:pval_t} for testing treatment effect quantiles with test statistics in \eqref{eq:combined_statistic_tail_prob_std} and \eqref{eq:combined_stat_max_sre} are approximately the same when the sample size is large.

The combined test statistic in \eqref{eq:combined_stat_max_sre} has a simpler form that avoids the use of the tail probability functions $G^{(h)}(\cdot)$, which typically lack simple closed-form expressions and require Monte Carlo approximations.
More importantly, the combined test statistic in \eqref{eq:combined_stat_max_sre} can facilitate the optimization for calculating the $p$-value in \eqref{eq:pval_t} through integer linear programming.
Finally, similar to the discussion after Theorem \ref{thm:cre_min_pval_combined_gaussian}, 
the combined test statistic in \eqref{eq:combined_stat_max_sre}, despite being motivated by the large-sample Gaussian approximation, always leads to finite-sample valid $p$-values under the SRE.
For completeness, we give the following theorem. 

\begin{theorem}\label{thm:combined_stat_max_sre}
    Under the SRE, the combined stratified rank sum statistic in \eqref{eq:combined_stat_max_sre} is distribution-free, and the $p$-value in \eqref{eq:pval_t} with $t$ equal to $\tilde{t}$ in \eqref{eq:combined_stat_max_sre} is valid for testing the null hypothesis $H_{k, c}^\treat$ in \eqref{eq:H_kc_t}, for any $1\le k \le n_\treat$ and $c \in \mathbb{R}$.
\end{theorem}

\begin{remark}\label{rmk:interval_scre}
    Corollary \ref{cor:interval_cre} on prediction and confidence intervals for treatment effect quantiles under the CRE is also true for the SRE, except that we require the test statistic to be distribution-free under the SRE. 
    We summarize the results below for conciseness. 
    For any $\alpha\in (0,1)$, 
    we can invert the test in Theorem \ref{thm:combined_stat_max_sre} to construct simultaneously valid $1-\alpha$ prediction intervals for treatment effect quantiles among treated units. 
    By switching the treatment labels and changing outcome signs, we can also construct simultaneously valid $1-\alpha$ prediction intervals for treatment effect quantiles among control units. 
    Combining these two sets of prediction intervals, we can further obtain simultaneously valid $1-2\alpha$ confidence intervals for treatment effect quantiles among all units.
\end{remark}

\subsubsection{Aggregation of the combined statistics within each stratum} \label{sec:comb2}

We now consider an alternative way of combining multiple rank statistics under the SRE. Specifically, we first perform the combination of rank statistics for each stratum separately, by considering
\begin{align}\label{eq:comb_per_stratum}
    \tilde{t}_s(\bs{z}_s, \bs{y}_s)
    =
    \max_{1\le h \le H} \frac{t_s^{(h)}(\bs{z}_s, \bs{y}_s) - \mu_s^{(h)}}{
    \text{sd}\{ \phi_s^{(h)}(U) \}
    },
    \quad (1\le s\le S)
\end{align}
where $\mu_s^{(h)}$ is defined as in \eqref{eq:mu_s},
$\text{sd}\{ \phi_s^{(h)}(U) \}$ denotes the standard deviation of $\phi_s^{(h)}(U)$ when $U\sim \text{Uniform}(0,1)$, and it simplifies to $\sqrt{(2\zeta-1)^{-1} - \zeta^{-2}}$ when $\phi_s^{(h)}(x) = x^{\zeta-1}$.
In other words, for each stratum we ask: which rank transformation best captures the treatment–control separation? We take the maximum of the standardized rank sum statistics within each stratum. A stratum in which treatment yields a few exceptionally large treated outcomes may favor large $\zeta$, whereas a stratum with broad, uniform gains may favor small $\zeta$.

The test statistic aggregated from all strata is then
\begin{align}\label{eq:comb_then_agg}
    \dtilde{t}(\bs{z}, \bs{y}) & = \sum_{s=1}^S w_s \tilde{t}_s(\bs{z}_s, \bs{y}_s),
\end{align}
where $w_s$s denote the weights for each stratum and can be chosen based on the two schemes in \eqref{eq:weighting_van_opt} and \eqref{eq:weighting_van_free}.
This ``combine-then-aggregate'' approach differs from the ``aggregate-then-combine'' approach in Section \ref{sec:comb1}: here, adaptation happens locally within each stratum before aggregation, potentially capturing heterogeneity across strata that a global combination may miss.

Below we explain the rationale for the choice of test statistics in \eqref{eq:comb_per_stratum} and \eqref{eq:comb_then_agg}.
We first consider the combination in \eqref{eq:comb_per_stratum} within any stratum $s$. When the number of stratum size $n_s$ is large, 
$\sigma_s^{(h)}$ defined as in \eqref{eq:sigma_s} is approximately $\{ n_{s0}/(n_{s1} n_s) \}^{1/2} \cdot \text{sd}\{ \phi_s^{(h)}(U) \}$. Consequently, the combined statistic for stratum $s$ in \eqref{eq:comb_per_stratum} is approximately 
\begin{align*}
    \tilde{t}_s(\bs{z}_s, \bs{y}_s)
    \approx
    \sqrt{\frac{n_{s0}}{n_{s1} n_s}} \cdot 
    \max_{1\le h \le H} \frac{t_s^{(h)}(\bs{z}_s, \bs{y}_s) - \mu_s^{(h)}}{
    \sigma_s^{(h)}
    }, 
\end{align*}
which, up to a constant scaling, is the maximum of the standardized rank statistics for stratum $s$ and has the same form as those combined statistics in \eqref{eq:combined_stat_max} and 
\eqref{eq:combined_stat_max_sre}. 
Therefore, we are essentially applying a similar combination as before, but within each stratum separately. 

We then consider the aggregation in \eqref{eq:comb_then_agg} from all strata. 
Suppose we apply the same rank transformation across all strata (i.e., $\phi_s^{(h)}$ equals some $\phi^{(h)}$ for all $s$) 
and 
there is a single rank statistic under consideration (i.e., $H=1$); then the aggregated statistic in \eqref{eq:comb_then_agg} simplifies to 
\begin{align*}
    \dtilde{t}(\bs{z}, \bs{y}) & = \sum_{s=1}^S w_s \frac{t_s^{(h)}(\bs{z}_s, \bs{y}_s) - \mu_s^{(h)}}{
    \text{sd}\{ \phi^{(h)}(U) \}
    }
    = 
    \frac{1}{\text{sd}\{ \phi^{(h)}(U) \}} \cdot \sum_{s=1}^S w_s t_s^{(h)}(\bs{z}_s, \bs{y}_s) - 
    \frac{\sum_{s=1}^S w_s \mu_s^{(h)}}{\text{sd}\{ \phi^{(h)}(U) \}}, 
\end{align*}
which has the same form as \eqref{eq:rank_sum_stratified_normalized} up to some constant scaling and shifting. 
Therefore, the weighting schemes discussed in Section \ref{sec:weighting} still apply here when $H=1$. 
Moreover, since the properties of the weighting schemes in \eqref{eq:weighting_van_opt} and \eqref{eq:weighting_van_free} hold for any choice of common rank transformation for all strata, it seems intuitive to apply the same weighting scheme when we use the combined rank statistic as in \eqref{eq:comb_per_stratum} for each stratum.
This also explains why we use $\text{sd}\{ \phi_s^{(h)}(U) \}$ instead of $\sigma_s^{(h)}$ for scaling in \eqref{eq:comb_per_stratum}, as this preserves a similar weighting scheme.

The test statistic in \eqref{eq:comb_then_agg} remains distribution-free under the SRE, enabling valid inference for quantiles of individual treatment effects. 
We summarize this in the following theorem. 

\begin{theorem}\label{thm:combine_then_aggregate}
    Theorem \ref{thm:combined_stat_max_sre} holds when we replace $\tilde{t}$ by $\dtilde{t}$, leading to a valid $p$-value for testing the null hypothesis $H_{k, c}^\treat$ in \eqref{eq:H_kc_t}, for any $1\le k \le n_\treat$ and $c \in \mathbb{R}$.
\end{theorem}

By the same logic as Remark \ref{rmk:interval_scre}, we can construct prediction or confidence intervals for quantiles of individual effects among treated, control or all units.

\subsection{Optimization of the $p$-value from the combination of multiple rank statistics}

The optimization for the $p$-value in Theorem \ref{thm:combine_then_aggregate} can be solved in almost the same way as \cite{SL22quantile}. 
Once we know there are at most $n_{s\treat}-k_s$ treated units in stratum $s$ with individual effects greater than $c$ for all $1\le s \le S$, 
the minimum value of the test statistic $\dtilde{t}(\bs{Z}, \bs{Y}-\bs{Z}\circ\bs{\delta})$ is achieved when the hypothesized effects $\bs{\delta}_s$ for units in each stratum $s$ equal $\bs{\xi}_{s, k,c}^\treat$, defined as in \eqref{eq:xi_kc_t} but only for units within stratum $s$.
Therefore, it suffices to find the worst-case configuration of the $k_s$s, which can be done in the same way as in \cite{SL22quantile}. 
We can obtain the minimum value of the test statistic exactly using integer linear programming or dynamic programming, or conservatively using linear programming or a greedy algorithm, with the latter yielding valid but more conservative $p$-values.

Below we focus on the optimization for the $p$-value in Theorem \ref{thm:combined_stat_max_sre}. 
From \eqref{eq:pval_t}, we need to minimize $\tilde{t}(\bs{Z}, \bs{Y} - \bs{Z}\circ \bs{\delta})$ subject to the constraint that $\bs{\delta} \in \mathcal{H}_{k, c}^\treat$, where $\tilde{t}(\cdot, \cdot)$ is the combined stratified rank sum statistic in \eqref{eq:combined_stat_max_sre}. 
By the same logic as the discussion in Section \ref{sec:single_rank_stat_stratified}, we can reformulate the optimization as a search for the worst-case allocation of treated units with large effects across strata. 
Specifically, 
analogous to \eqref{eq:H_kc_t},  
define $\mathcal{H}_{s,k_s,c}^\treat \subset \mathbb{R}^{n_s}$ to denote the set of all possible values of individual treatment effect vectors for units in stratum $s$ such that there are at most $n_{s\treat} - k_s$ treated units having individual effects greater than $c$ in the stratum, for all $1\le s \le S$ and $0\le k_s \le n_{s\treat}$. 
For simplicity, we will write hypothesized individual effects as $\bs{\delta} = (\bs{\delta}_1, \bs{\delta}_2, \ldots, \bs{\delta}_S)$, with $\bs{\delta}_s \in \mathbb{R}^{n_s}$ representing the hypothesized individual effects for stratum $s$. 
The minimum of $\tilde{t}(\bs{Z}, \bs{Y} - \bs{Z}\circ \bs{\delta})$ under the constraint that $\bs{\delta} = (\bs{\delta}_1, \ldots, \bs{\delta}_S) \in \mathcal{H}_{k, c}^\treat$
is equivalent to
\begin{equation}\label{eq:obj_given_ks}
    \inf_{\bs{\delta}
    \in \prod_{s=1}^S \mathcal{H}_{s,k_s,c}^\treat}\tilde{t}(\bs{Z}, \bs{Y} - \bs{Z}\circ \bs{\delta})
    =
    \inf_{\bs{\delta}_s \in \mathcal{H}_{s,k_s,c}^\treat \forall s}\max_{1\le h\le H} \frac{
        \sum_{s=1}^S w_s^{(h)} t_s^{(h)}(\bs{Z}_s, \bs{Y}_s - \bs{Z}_s\circ \bs{\delta}_s)-\mu^{(h)}
        }{\sigma^{(h)}},
\end{equation}
subject to the constraint that
\begin{align}\label{eq:constraint_ks}
    0\le k_s \le n_{s\treat} \text{ for all } s,
    \text{ and }
    \sum_{s=1}^S k_s = k.
\end{align}
In other words, to test a hypothesis like $H_{k,c}^\treat$, a researcher must decide how to allocate the $k$ units with effects exceeding $c$ across the $S$ strata. The constraint says we can assign between zero and all treated units in each stratum to have large effects, as long as the total across strata equals $n_1- k$. For each allocation, we compute the worst-case test statistic; then we search over all allocations to find the one most favorable to the null.

Given any fixed $(k_1, \ldots, k_S)$ satisfying \eqref{eq:constraint_ks}, the optimization in \eqref{eq:obj_given_ks} has a closed-form solution. 
This is because for each stratum $1\le s\le S$, the rank sum statistic $t_s^{(h)}(\bs{Z}_s, \bs{Y}_s - \bs{Z}_s\circ \bs{\delta}_s)$ achieves its infimum at the same configuration of $\bs{\delta}_s\in \mathcal{H}_{s,k_s,c}^\treat$ for all $h$. 
This follows by the same logic as the discussion on rank sum statistic under the CRE in Section \ref{sec:single_test}. 
Specifically, for all $h$ and $s$, 
$\inf_{\bs{\delta}_s\in \mathcal{H}_{s,k_s,c}^\treat}t_s^{(h)}(\bs{Z}_s, \bs{Y}_s - \bs{Z}_s\circ \bs{\delta}_s)
= 
t_s^{(h)}(\bs{Z}_s, \bs{Y}_s - \bs{Z}_s\circ \bs{\xi}_{s,k,c}^\treat),
$
where $\bs{\xi}_{s,k,c}^\treat$ defined analogously as in \eqref{eq:xi_kc_t} but only for units in stratum $s$. 
For any $1\le h\le H$ and $1\le s\le S$, let 
\begin{align*}
    t_s^{(h)}(k,c) \equiv 
    t_s^{(h)}(\bs{Z}_s, \bs{Y}_s - \bs{Z}_s\circ \bs{\xi}_{s,k,c}^\treat)
    & = \inf_{\bs{\delta}_s\in \mathcal{H}_{s,k_s,c}^\treat}t_s^{(h)}(\bs{Z}_s, \bs{Y}_s - \bs{Z}_s\circ \bs{\delta}_s)
\end{align*}
We can then equivalently write \eqref{eq:obj_given_ks} as 
\begin{align}\label{eq:obj_given_ks_equiv} 
     \inf_{\bs{\delta} \in \prod_{s=1}^S \mathcal{H}_{s,k_s,c}^\treat}\tilde{t}(\bs{Z}, \bs{Y} - \bs{Z}\circ \bs{\delta})
    & = 
    \max_{1\le h\le H} \frac{
        \sum_{s=1}^S w_s^{(h)} t_s^{(h)}(k,c)-\mu^{(h)}
        }{\sigma^{(h)}}.
\end{align}

To formulate the optimization of \eqref{eq:obj_given_ks} or equivalently \eqref{eq:obj_given_ks_equiv} under the constraint \eqref{eq:constraint_ks} as an integer linear programming problem,
we introduce binary variables $x_{sj}\in \{0,1\}$ for all $1\le s\le S$ and $0\le j\le n_{s1}$, and use them to represent the values of the $k_s$s. 
For each $s$, $k_s = m$ if and only if $x_{sm} = 1$ and $x_{sj} = 0$ for $j\ne m$. 
The optimization of \eqref{eq:obj_given_ks_equiv} under the constraint \eqref{eq:constraint_ks} can then be reformulated as
\begin{align}\label{eq:ILP}
    \min \quad & {t}_\star
    \\
    \text{subject to} \quad
    &
    \sigma^{(h)} {t}_\star \ge
    \sum_{s=1}^S w_s^{(h)} \sum_{j=0}^{n_{s1}} x_{sj} t_s^{(h)}(j, c)-\mu^{(h)} \text{ for } 1\le h \le H
    \nonumber
    \\
    & \sum_{j=0}^{n_{s\treat}} x_{sj} = 1 \text{ for } 1\le s \le S,
    \qquad
    \sum_{s=1}^S \sum_{j=0}^{n_{s\treat}} x_{sj} j = k, 
    \nonumber
    \\
    &x_{sj} \in \{0, 1\} \text{ for } 1 \le s \le S \text{ and } 0\le j \le n_{s\treat}.
    \nonumber
\end{align}

In other words, we introduce binary decision variables $x_{sj}$ that equal 1 if stratum $s$ contributes exactly $n_{s1}- j$ units with effects exceeding $c$, and 0 otherwise. The first constraint ensures that the objective $t_\star$ is at least as large as each individual standardized statistic---making $t_\star$ the maximum. The second constraint says each stratum picks exactly one allocation. The third says the allocations sum to $k$ across strata. 
We can also relax the integer constraint in \eqref{eq:ILP} by allowing $x_{sj} \in [0,1]$ for all $s$ and $j$, in which case the optimization becomes a linear programming problem  that is solvable in polynomial time.
Moreover, this relaxation also leads to valid, although more conservative, $p$-values for testing treatment effect quantiles.

\section{Simulation studies}\label{sec:simu}

To save space, we relegate the details of the simulation studies to the supplementary material and summarize the main findings here. 
First, under the CRE, different individual rank sum statistics are preferred for different data-generating processes and different quantiles of interest, whereas the combined statistic consistently performs close to the best individual one. 
Second, under the SRE, the weighting scheme in \eqref{eq:weighting_van_opt} outperforms that in \eqref{eq:weighting_van_free}, and the ``combine-then-aggregate'' approach in Section \ref{sec:comb2} outperforms the ``aggregate-then-combine'' approach in Section \ref{sec:comb1}.

\section{Application: Detecting the causal effects of teacher training}\label{sec:app}

\subsection{Analyses under a CRE} \label{sec:credata}

We revisit the teacher training experiment \citep{HSMHSD10}   previously analyzed in
\citet{CDLM21quantile} and \citet{ZL24quantile}, where a group of fourth grade
teachers were randomly assigned to treatment and control. Treated teachers
participated in professional development courses focused on teaching about
electric circuits. The outcome of interest is change in teachers' content knowledge of
electric circuits, measured by pre- and post-course
assessments. We use the same dataset as in \citet{CDLM21quantile} and
\citet{ZL24quantile} following their preprocessing steps, and here analyze it as a
CRE with $164$ treated teachers and $69$ control teachers. As a benchmark, we
first apply the method of \citet{ZL24quantile} using the Stephenson rank sum
statistic with $\zeta = 2$, $6$, $10$, and $30$, respectively. We then apply
the proposed approach combining these four rank statistics. 

\definecolor{myred}{RGB}{218,72,96}
\definecolor{myblue}{RGB}{30,139,226}
\definecolor{mycyan}{RGB}{36,221,224}
\definecolor{mygreen}{RGB}{85,202,70}
\definecolor{mypurple}{RGB}{182,100,218}

\begin{figure}[htbp]
    \centering
    
    \begin{subfigure}[b]{0.4\linewidth}
    \includegraphics[width=\linewidth]{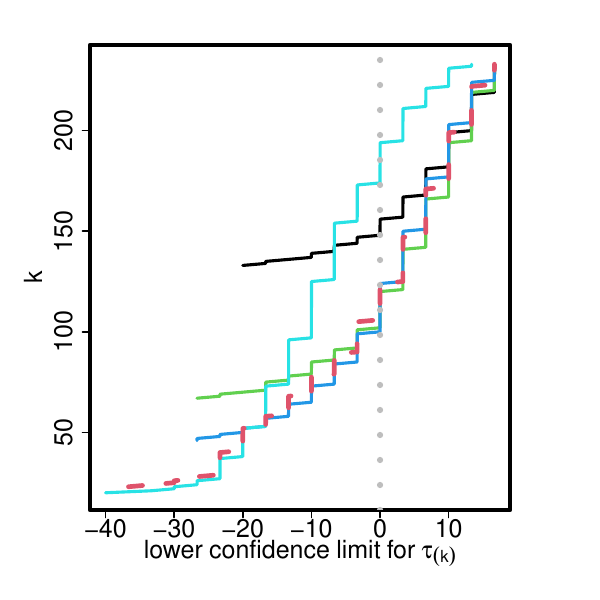}
    \end{subfigure}
    \begin{subfigure}[b]{0.1\linewidth}
        \begin{tikzpicture}
    \node[draw, rounded corners, inner sep = 2pt] (legend){
    \begin{tikzpicture}[font = \tiny]
        \matrix[column sep = 1mm]{
        \draw[black, thick] (0,0) -- (0.3,0); &\node {Stephenson with $\zeta = 2$}; \\
        \draw[mygreen, thick] (0,0) -- (0.3,0); &\node {Stephenson with $\zeta = 6$}; \\
        \draw[myblue, thick] (0,0) -- (0.3,0); &\node {Stephenson with $\zeta = 10$}; \\
        \draw[mycyan, thick] (0,0) -- (0.3,0); &\node {Stephenson with $\zeta = 30$}; \\
        \draw[myred, thick, dashed] (0,0) -- (0.3,0); &\node {Combined}; \\
        };
    \end{tikzpicture}
    };
    \end{tikzpicture}
    \vspace{1.8cm}
    \end{subfigure}

    \caption{
    Lower confidence bounds for treatment effect quantiles in the education experiment under the CRE.
    Solid lines denote rank sum statistics with Stephenson rank transformations ($\zeta = 2, 6, 10, 30$), and the dashed line denotes the combined method in Corollary \ref{cor:cre_min_pval}.
    }
    \label{fig:cre_elec}
\end{figure}
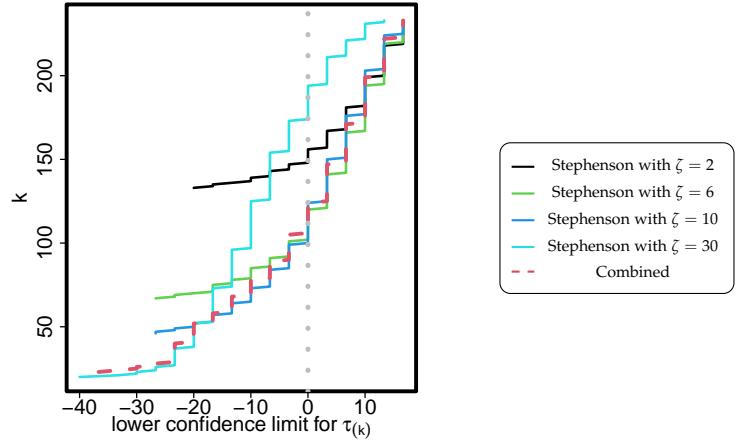

We present the $90\%$ simultaneous lower confidence bounds for all effect
quantiles $\tau_{(k)}$s from each method in Figure
\ref{fig:cre_elec}. The figure shows that methods based on a single
Stephenson rank statistic vary in performance across quantiles, whereas our
combined method consistently performs close to the best individual method. The performance of individual methods differs substantially across choices of the
tuning parameter $\zeta$ in both the number of informative quantiles and the tightness of lower confidence limits. For the
method of \citet{ZL24quantile}, the intervals are uninformative (i.e. have lower limits at $-\infty$) for $k \leq 133, 66, 45$ and $20$
when $\zeta = 2, 6, 10,$ and $30$, respectively. Our combined method is uninformative only for $k \leq 23$, close to the best choice ($\zeta = 30$) and substantially better than the others. Furthermore, its limits are nearly identical to those from the best $\zeta$ at each informative quantile.

What is the causal effect of the training? The number of lower confidence limits for $\tau_{(k)}$s exceeding zero indicates how many teachers' effects can be confidently bounded away from zero---equivalently, how many benefited from the professional development. Figure \ref{fig:cre_elec}(a) shows that this conclusion depends crucially on the test statistic. With $\zeta = 6$ or $\zeta = 10$, the $90\%$ lower bounds imply that at least $52\%$ or $50\%$ of teachers benefited, respectively. With $\zeta = 2$ or $\zeta = 30$, however, the same procedure implies only $34\%$ or $21\%$ benefited. Our combined method 
implies 
at least $50\%$ benefited---matching the most informative individual choices ($\zeta = 6$ and $10$) without requiring the analyst to select $\zeta$ in advance. The combined procedure avoids the sensitivity of single-rank methods to the tuning parameter and adapts to the most informative rank transformation.

\subsection{Analyses under a SRE} \label{sec:sredata}

We revisit the dataset analyzed in Section \ref{sec:credata}, but now treat it as an SRE, with strata corresponding to research sites across different U.S. states.\footnote{Randomization was conducted within research sites. There are $7$ strata, with sizes ranging from $15$ to $54$, and the proportion of treated units within each stratum varies from $60\%$ to $81\%$. For further details, see \citet{HSMHSD10}.} As a benchmark, we first apply the method in \citet{ZL24quantile} based on a single stratified rank sum statistic in \eqref{eq:rank_sum_stratified_normalized}, 
using a polynomial rank transformation with $\zeta = 2$, $6$, and $10$, respectively, for normalized ranks across all strata.
We then apply our proposed approaches combining these rank statistics as detailed in Sections \ref{sec:comb1} and \ref{sec:comb2}, where we consider both integer programming and linear programming for exactly or conservatively calculating the $p$-values. 
For the weighting scheme, we adopt the choice in \eqref{eq:weighting_van_opt}, given its superior performance in the simulation studies reported in the supplementary material.

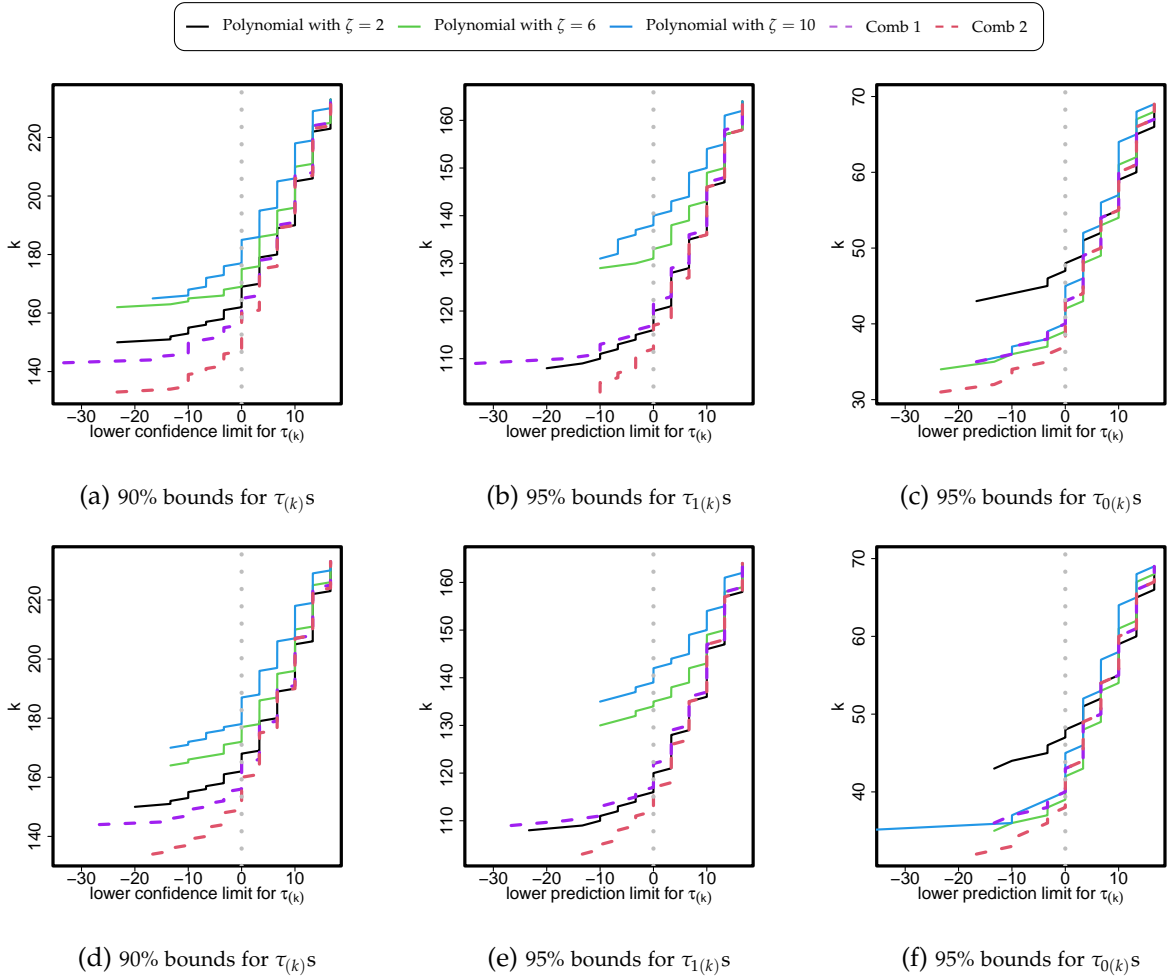
\begin{figure}[htbp]
    \centering
    \begin{tikzpicture}
    \node[draw, rounded corners, inner sep = 2pt] (legend){
    \begin{tikzpicture}[font = \tiny]
        \matrix[column sep = 1mm]{
        \draw[black, thick] (0,0) -- (0.3,0); &\node {Polynomial with $\zeta = 2$}; &
        \draw[mygreen, thick] (0,0) -- (0.3,0); &\node {Polynomial with $\zeta = 6$}; &
        \draw[myblue, thick] (0,0) -- (0.3,0); &\node {Polynomial with $\zeta = 10$};  &\draw[mypurple, thick, dashed] (0,0) -- (0.3,0); &\node {Comb 1}; 
        &\draw[myred, thick, dashed] (0,0) -- (0.3,0); 
        &\node {Comb 2}; \\
        };
    \end{tikzpicture}
    };
    \end{tikzpicture}
    
    \begin{subfigure}[b]{0.33\linewidth}
    \includegraphics[width=\linewidth]{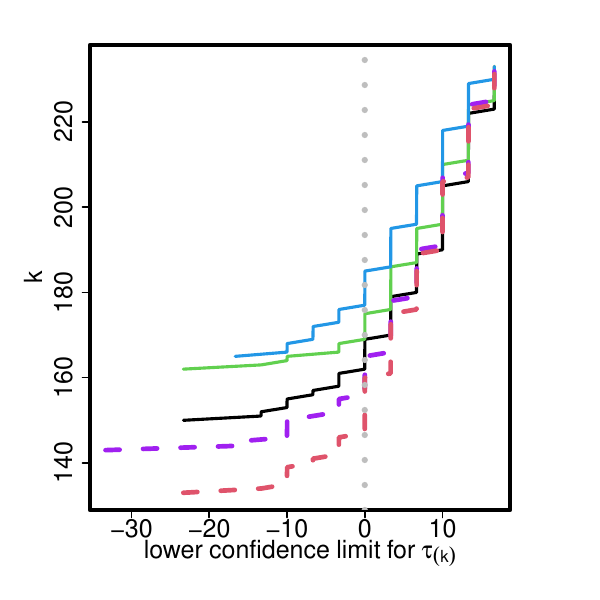}
    \caption{\scriptsize $90\%$ bounds for $\tau_{(k)}$s
    }
    \end{subfigure}%
    \begin{subfigure}[b]{0.33\linewidth}
    \includegraphics[width=\linewidth]{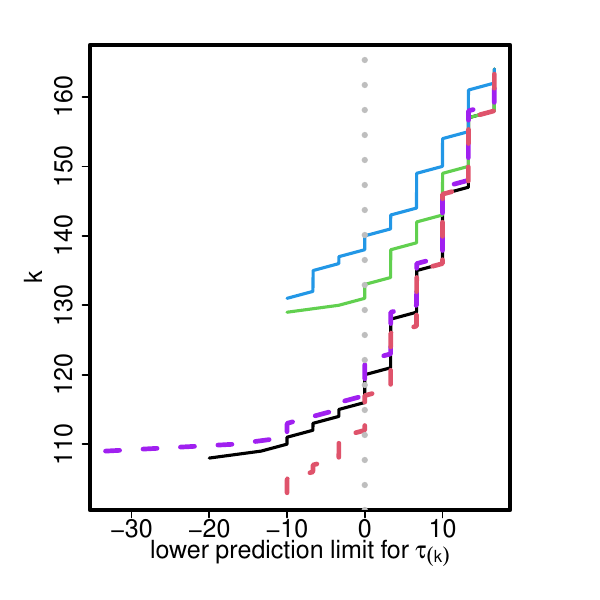}
    \caption{\scriptsize $95\%$ bounds for $\tau_{1(k)}$s 
    }
    \end{subfigure}%
    \begin{subfigure}[b]{0.33\linewidth}
    \includegraphics[width=\linewidth]{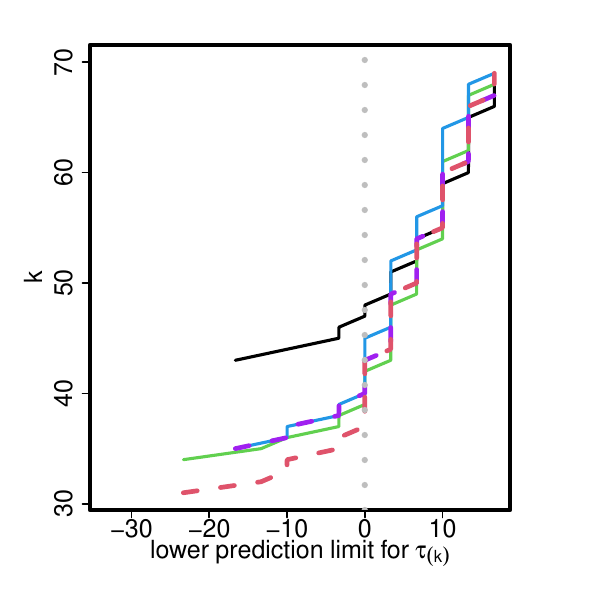}
    \caption{\scriptsize $95\%$ bounds for $\tau_{0(k)}$s 
    }
    \end{subfigure}
    \begin{subfigure}[b]{0.33\linewidth}
    \includegraphics[width=\linewidth]{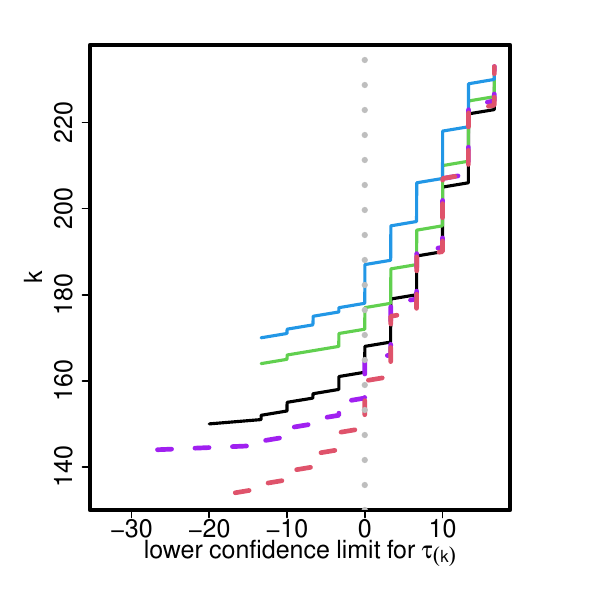}
    \caption{\scriptsize $90\%$ bounds for $\tau_{(k)}$s 
    }
    \end{subfigure}%
    \begin{subfigure}[b]{0.33\linewidth}
    \includegraphics[width=\linewidth]{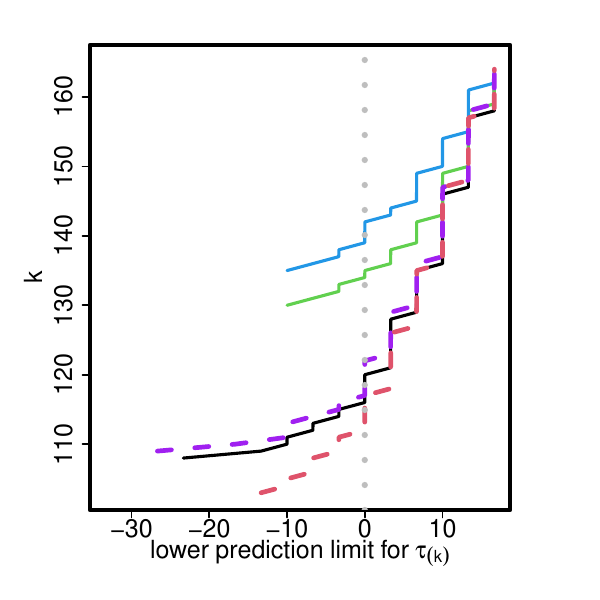}
    \caption{\scriptsize $95\%$ bounds for $\tau_{1(k)}$s 
    }
    \end{subfigure}%
    \begin{subfigure}[b]{0.33\linewidth}
    \includegraphics[width=\linewidth]{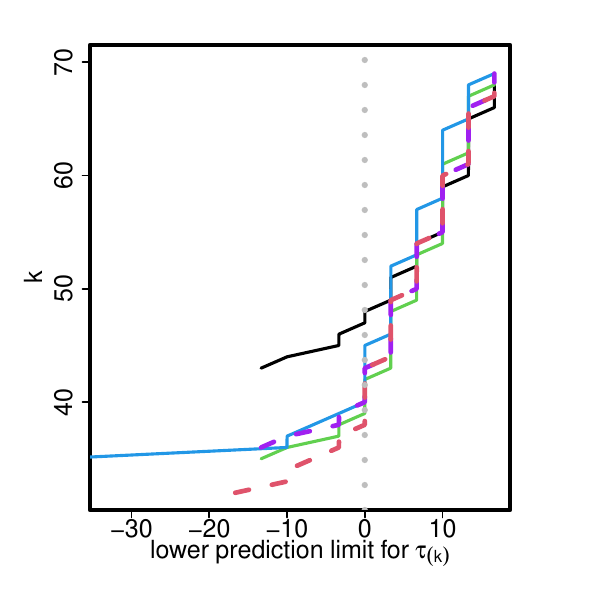}
    \caption{\scriptsize $95\%$ bounds for $\tau_{0(k)}$s 
    }
    \end{subfigure}

    \caption{Confidence and prediction bounds for treatment effect quantiles in  education experiment under the SRE. Panels (b) and (c) show the $95\%$ simultaneous lower prediction bounds for treated and control units using integer linear programming (exact $p$-values);
    panel (a) shows the corresponding $90\%$ simultaneous lower confidence bounds among all units obtained by combining (b) and (c). 
    Panels (d)--(f) are analogous but use linear programming (conservative $p$-values).
    Solid lines denote single stratified rank statistic with polynomial rank transformation ($\zeta$ = $2$, $6$, $10$); dashed lines denote the two combination methods (Comb 1: Section \ref{sec:comb1}; Comb 2: Section \ref{sec:comb2}).
    }
    \label{fig:sre_elec}
\end{figure}

Figure \ref{fig:sre_elec} shows the $90\%$ lower confidence bounds for quantiles of individual treatment effects across all units, along with the associated $95\%$ lower prediction bounds for quantiles of individual treatment effects among treated and control units, 
computed using both integer and relaxed linear programming. 
Similar to the analysis under CRE in Section \ref{sec:credata}, our two combined methods perform close to the best individual method in both the number of informative quantiles and tightness of the resulting bounds. 
How many teachers benefited from the professional development course?
From Figure \ref{fig:sre_elec}(a), the $90\%$ lower bounds from the combined methods imply that at least $31\%$ and $34\%$ of teachers benefited, respectively. 
By contrast, individual methods using polynomial rank transformations with $\zeta = 2$, $6$, and $10$ imply only $28\%$, $27\%$, and $23\%$ of teachers benefited, respectively.

Moreover, as illustrated in Figure \ref{fig:sre_elec}(a), the combined method outperforms all individual methods for confidence bounds on all units when $k$ is up to about $180$. This improvement arises because inference for treated and control units may favor different rank transformations. 
For example, Figure \ref{fig:sre_elec}(b) suggests that for prediction bounds on $\tau_{\treat(k)}$s among treated units, smaller tuning parameter ($\zeta = 2$) generally performs better, whereas Figure \ref{fig:sre_elec}(c) indicates that for prediction bounds on $\tau_{\control(k)}$s among control units, larger $\zeta$ ($\zeta = 6$ or $10$) yields more informative bounds when $k \leq 50$. 
As a result, after aggregating the prediction bounds from the treated and control groups, the combined method integrates the strengths of different rank transformations and is more informative than any individual method for $k$ up to about $180$.
Moreover, the combined method in Section \ref{sec:comb2} generally outperforms that in Section \ref{sec:comb1} across all panels of \ref{fig:sre_elec}, especially for lower quantiles. This aligns with our simulation results reported in the supplementary material.

Finally, comparing integer linear programming (Figure \ref{fig:sre_elec}(a)-(c)) and linear programming (Figure \ref{fig:sre_elec}(d)-(f)) shows that relaxation leads to only minor losses. In practice, this suggests that linear programming can be recommended for large-scale problems, as it offers faster computation with only a small loss in information.

\section{Conclusion and Discussion}\label{sec:discussion}

How many people actually benefited? This question motivates much of program evaluation. An average treatment effect
cannot distinguish a world where everyone gains modestly from one where a few
gain dramatically and others not at all. Yet these worlds demand different
policy responses.

We have developed methods that let researchers characterize the distribution of
individual treatment effects without committing in advance to a single rank-based test statistic. The key insight is that different rank statistics excel at
detecting effects in different parts of the distribution and under different data-generating processes: some reveal dramatic responders, others detect widespread modest gains. By combining
multiple rank statistics while maintaining finite-sample validity, our approach
adapts to whichever pattern the data contain. The researcher need not guess
correctly beforehand.

The practical stakes are substantial. In the teacher training experiment
analyzed as a completely randomized experiment for illustration of the method,
an unlucky choice of tuning parameter would suggest only 21\% of teachers
benefited; a luckier choice would suggest 52\%. This is not a minor
methodological nuance---it is the difference between a program that helps a
small minority and one that helps a majority. Our combined test, requiring no
such choice, yields bounds implying roughly half the teachers benefited, close
to the most informative individual method. The procedure
extracts nearly the full information available without demanding prescience
from the analyst.

We focused on lower bounds for effect quantiles, constructed by inverting
one-sided tests. Although not shown, upper bounds
follow by reversing outcome signs, and two-sided intervals could arise from
combining lower and upper bounds. The methods we developed here apply to both
completely randomized and stratified randomized experiments; for stratified
designs, we developed normalized rank transformations and principled weighting
schemes that prevent large strata from overwhelming small ones, and we showed
that linear programming relaxations sacrifice little precision while
substantially reducing computation.

Several directions remain open. The optimal choice of which rank statistics to
combine, and with what weights, likely depends on features of the effect
distribution that are unknown in practice; theoretical guidance here would be
valuable. Our methods assume the experiment was randomized; extending them to
observational studies with potentially confounded treatment assignment would
broaden their applicability \citep{Rosenbaum02a, SL22quantile, WL23,
ChenLiZhang2024}. Settings with missing or censored outcomes present additional
challenges \citep{li2021randomization, Zhaomissing24, li2025randomization,
Heng29082025}. We leave these extensions to future work.

\bibliographystyle{apalike}
\bibliography{reference.bib}

\begin{thebibliography}{}

\bibitem[Athey and Wager, 2021]{athey2021observational}
Athey, S. and Wager, S. (2021).
\newblock Policy learning with observational data.
\newblock {\em Econometrica}, 89(1):133--161.

\bibitem[Bickel and Freedman, 1984]{Bickel1984}
Bickel, P.~J. and Freedman, D.~A. (1984).
\newblock {Asymptotic Normality and the Bootstrap in Stratified Sampling}.
\newblock {\em The Annals of Statistics}, 12:470 -- 482.

\bibitem[Caughey et~al., 2023]{CDLM21quantile}
Caughey, D., Dafoe, A., Li, X., and Miratrix, L. (2023).
\newblock Randomization inference beyond the sharp null: Bounded null
  hypotheses and quantiles of individual treatment effects.
\newblock {\em Journal of the Royal Statistical Society, Series B (Statistical
  Methodology)}, 85:1471--1491.

\bibitem[Chen and Li, 2026]{ZL24quantile}
Chen, Z. and Li, X. (2026).
\newblock Enhanced inference for distributions and quantiles of individual
  treatment effects in various experiments.
\newblock {\em Journal of the American Statistical Association}, page inpress.

\bibitem[Chen et~al., 2024]{ChenLiZhang2024}
Chen, Z., Li, X., and Zhang, B. (2024).
\newblock The role of randomization inference in unraveling individual
  treatment effects in early phase vaccine trials.
\newblock {\em Statistical Communications in Infectious Diseases}, 16:20240001.

\bibitem[Fisher, 1935]{Fisher:1935}
Fisher, R.~A. (1935).
\newblock {\em The {D}esign of {E}xperiments, 1st Edition}.
\newblock Edinburgh, London: Oliver and Boyd.

\bibitem[H{\'a}jek, 1960]{hajek1960limiting}
H{\'a}jek, J. (1960).
\newblock Limiting distributions in simple random sampling from a finite
  population.
\newblock {\em Publications of the Mathematics Institute of the Hungarian
  Academy of Science}, 5:361--374.

\bibitem[Heckman et~al., 1997]{heckman1997}
Heckman, J.~J., Smith, J., and Clements, N. (1997).
\newblock Making the most out of programme evaluations and social experiments:
  Accounting for heterogeneity in programme impaces.
\newblock {\em The Review of Economic Studies}, 64(4):487--535.

\bibitem[Heller et~al., 2010]{HSMHSD10}
Heller, J.~L., Shinohara, M., Miratrix, L., Hesketh, S.~R., and Daehler, K.~R.
  (2010).
\newblock Learning science for teaching: Effects of professional development on
  elementary teachers, classrooms, and students.
\newblock {\em Proceedings from Society for Research on Educational
  Effectiveness.}

\bibitem[Heng et~al., 2025]{Heng29082025}
Heng, S., Zhang, J., and Feng, Y. (2025).
\newblock Design-based causal inference with missing outcomes: Missingness
  mechanisms, imputation-assisted randomization tests, and covariate
  adjustment.
\newblock {\em Journal of the American Statistical Association}, in press.

\bibitem[Imai and Ratkovic, 2013]{imai2013}
Imai, K. and Ratkovic, M. (2013).
\newblock Estimating treatment effect heterogeneity in randomized program
  evaluation.
\newblock {\em The Annals of Applied Statistics}, 7(1):443--470.

\bibitem[Koenker, 2017]{koenker2017quantile}
Koenker, R. (2017).
\newblock Quantile regression: 40 years on.
\newblock {\em Annual review of economics}, 9(1):155--176.

\bibitem[Li and Ding, 2017]{LD2017}
Li, X. and Ding, P. (2017).
\newblock General forms of finite population central limit theorems with
  applications to causal inference.
\newblock {\em Journal of the American Statistical Association},
  112:1759--1769.

\bibitem[Li et~al., 2025]{li2025randomization}
Li, X., Sheng, P., and Yu, Z. (2025).
\newblock Randomization inference with sample attrition.
\newblock {\em arXiv preprint arXiv:2507.00795}.

\bibitem[Li and Small, 2022]{li2021randomization}
Li, X. and Small, D.~S. (2022).
\newblock Randomization-based test for censored outcomes: A new look at the
  logrank test.
\newblock {\em Statistical Science}, page To appear.

\bibitem[Liu and Yang, 2020]{LY2020stratifiedadj}
Liu, H. and Yang, Y. (2020).
\newblock {Regression-adjusted average treatment effect estimates in stratified
  randomized experiments}.
\newblock {\em Biometrika}, 107:935--948.

\bibitem[Manski, 2004]{manski2004}
Manski, C.~F. (2004).
\newblock Statistical treatment rules for heterogeneous populations.
\newblock {\em Econometrica}, 72(4):1221--1246.

\bibitem[Neyman, 1923]{Neyman:1923}
Neyman, J. (1923).
\newblock {On the application of probability theory to agricultural
  experiments. Essay on principles (with discussion). Section 9 (translated).
  reprinted ed.}
\newblock {\em Statistical Science}, 5:465--472.

\bibitem[Nie and Wager, 2021]{NieWager2020}
Nie, X. and Wager, S. (2021).
\newblock {Quasi-oracle estimation of heterogeneous treatment effects}.
\newblock {\em Biometrika}, 108(2):299--319.

\bibitem[Puri, 1965]{puri1965combination}
Puri, M.~L. (1965).
\newblock On the combination of independent two sample tests of a general
  class.
\newblock {\em Revue de l'Institut International de Statistique}, pages
  229--241.

\bibitem[Qu et~al., 2025]{qu2024randomization}
Qu, T., Du, J., and Li, X. (2025).
\newblock Randomization-based z-estimation for evaluating average and
  individual treatment effects.
\newblock {\em Biometrika}, 112(2):1--9.

\bibitem[Rosenbaum, 2002]{Rosenbaum02a}
Rosenbaum, P.~R. (2002).
\newblock {\em Observational Studies}.
\newblock Springer, New York, 2 edition.

\bibitem[Rosenbaum, 2007]{Rosenbaum2007dramatic}
Rosenbaum, P.~R. (2007).
\newblock Confidence intervals for uncommon but dramatic responses to
  treatment.
\newblock {\em Biometrics}, 63:1164--1171.

\bibitem[Rosenbaum and Silber, 2008]{rosenbaum2008aberrant}
Rosenbaum, P.~R. and Silber, J.~H. (2008).
\newblock Aberrant effects of treatment.
\newblock {\em Journal of the American Statistical Association},
  103(481):240--247.

\bibitem[Rubin, 1974]{Rubin:1974}
Rubin, D.~B. (1974).
\newblock Estimating causal effects of treatments in randomized and
  nonrandomized studies.
\newblock {\em Journal of Educational Psychology}, 66:688--701.

\bibitem[Shi and Li, 2024]{shi2024some}
Shi, L. and Li, X. (2024).
\newblock Some theoretical foundations for the design and analysis of
  randomized experiments.
\newblock {\em Journal of Causal Inference}, 12(1).

\bibitem[Stephenson and Ghosh, 1985]{Stephenson85}
Stephenson, W.~R. and Ghosh, M. (1985).
\newblock Two sample nonparametric tests based on subsamples.
\newblock {\em Communications in Statistics - Theory and Methods},
  14:1669--1684.

\bibitem[Su and Li, 2024]{SL22quantile}
Su, Y. and Li, X. (2024).
\newblock {Treatment effect quantiles in stratified randomized experiments and
  matched observational studies}.
\newblock {\em Biometrika}, 111(1):235--254.

\bibitem[Tian et~al., 2014]{tian2014}
Tian, L., Alizadeh, A.~A., Gentles, A.~J., and Tibshirani, R. (2014).
\newblock A simple method for estimating interactions between a treatment and a
  large number of covariates.
\newblock {\em Journal of the American Statistical Association},
  109:1517--1532.

\bibitem[van Elteren, 1960]{van1960combination}
van Elteren, P.~H. (1960).
\newblock On the combination of independent two sample tests of wilcoxon.
\newblock {\em Bulletin of the Institute of International Statistics},
  37:351--361.

\bibitem[Wu and Li, 2025]{WL23}
Wu, D. and Li, X. (2025).
\newblock Sensitivity analysis for quantiles of hidden biases in matched
  observational studies.
\newblock {\em Journal of the American Statistical Association},
  120:1657--1668.

\bibitem[Zhang et~al., 2012]{zhang2012}
Zhang, B., Tsiatis, A.~A., Davidian, M., Zhang, M., and Laber, E. (2012).
\newblock Estimating optimal treatment regimes from a classification
  perspective.
\newblock {\em Stat}, 1(1):103--114.

\bibitem[Zhao et~al., 2024]{Zhaomissing24}
Zhao, A., Ding, P., and Li, F. (2024).
\newblock Covariate adjustment in randomized experiments with missing outcomes
  and covariates.
\newblock {\em Biometrika}, 111:1413--1420.

\end{thebibliography}

\newpage
\begin{center}
\textbf{\Large Supplemental Material to ``Randomization Tests for Distributions of Individual Treatment Effects using Multiple Rank Statistics''}
\end{center}
\setcounter{equation}{0}
\setcounter{theorem}{0}
\setcounter{corollary}{0}
\setcounter{proposition}{0}
\setcounter{figure}{0}
\setcounter{table}{0}
\setcounter{page}{1}
\setcounter{section}{0}

\makeatletter

\renewcommand{\theproposition}{A\arabic{proposition}}
\renewcommand{\thetheorem}{A\arabic{theorem}}
\renewcommand{\thecorollary}{A\arabic{corollary}}

\renewcommand{\thepage}{A\arabic{page}}
\renewcommand{\thesection}{A\arabic{section}}
\renewcommand{\theequation}{S\arabic{equation}}
\renewcommand{\thefigure}{A\arabic{figure}}
\renewcommand{\bibnumfmt}[1]{[S#1]}
\renewcommand{\citenumfont}[1]{S#1}

\section{Simulation studies}\label{sec:numerical}

\subsection{Simulation under the CRE}\label{subsec:cre_sim}

In this subsection, we conduct simulation to evaluate the performance of our method under the CRE. 
We consider a CRE with sample size $n =  100$, where $n/2$ units are randomly assigned to treatment and control. 
We generate potential outcomes from the following model:
\begin{align}\label{eq:cresimmodel}
  Y_i(0) \stackrel{\text{i.i.d.}}{\sim} \mathcal{N}(0, 1), \ Y_i(1) \stackrel{\text{i.i.d.}}{\sim} \mathcal{N}(2, \sigma^2), \ \ \text{where} \ \sigma \in \{ 0.2, \  0.5, \  1, \ 2, \  5 \}.
\end{align}
In other words, every unit benefits from treatment by 2 points on average, but the spread of treatment outcomes varies. When $\sigma < 1$, treated units cluster tightly around 2, indicating compressed variation under treatment. When $\sigma > 1$, treatment amplifies differences, with some units benefiting substantially and others only modestly. 
These scenarios assess whether the combined method adapts to different patterns of heterogeneity.

For each simulation iteration, we generate new potential outcomes from \eqref{eq:cresimmodel}, together with a new treatment assignment vector.
In this way, we will investigate the average performance of different approaches when data are generated from \eqref{eq:cresimmodel}. 
We compare existing approaches, which use a single Stephenson rank sum statistic with $\zeta = 2, 6, 10, 30$, to our approach which combines these four rank sum statistics as discussed in Corollary \ref{cor:cre_min_pval}.

\begin{figure}[htb]
    
    \centering
    \begin{subfigure}[b]{0.32\linewidth}
    \includegraphics[width=\linewidth]{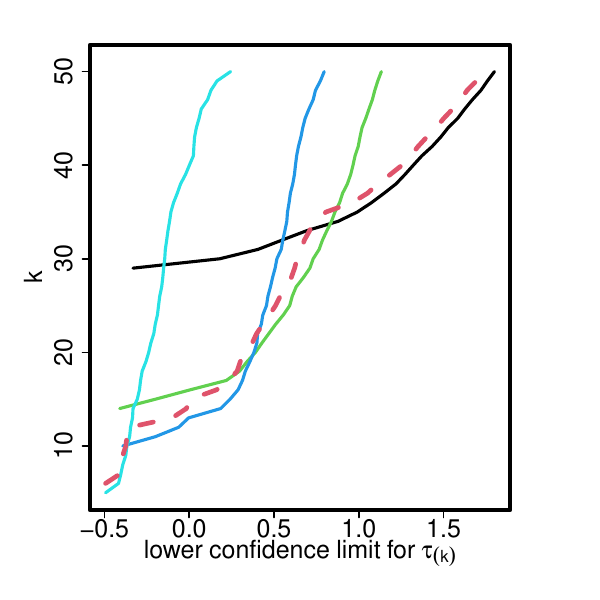}
    \caption{\small $\sigma = 0.2$}
    \end{subfigure}
    \begin{subfigure}[b]{0.32\linewidth}
    \includegraphics[width=\linewidth]{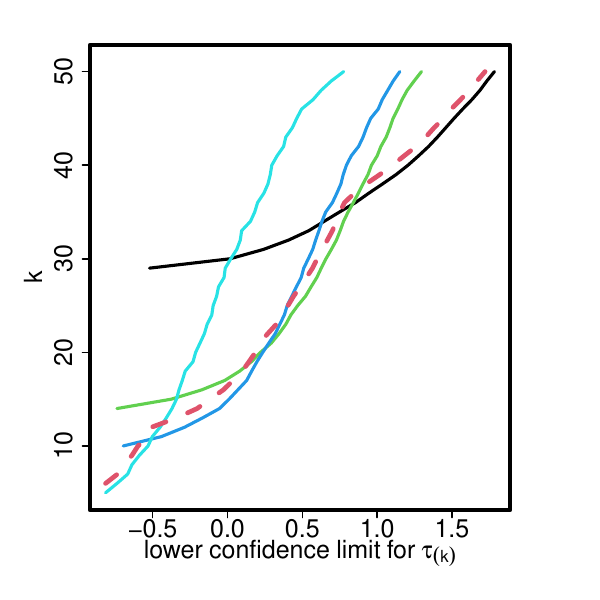}
    \caption{\small $\sigma = 0.5$}
    \end{subfigure}
    \begin{subfigure}[b]{0.32\linewidth}
    \includegraphics[width=\linewidth]{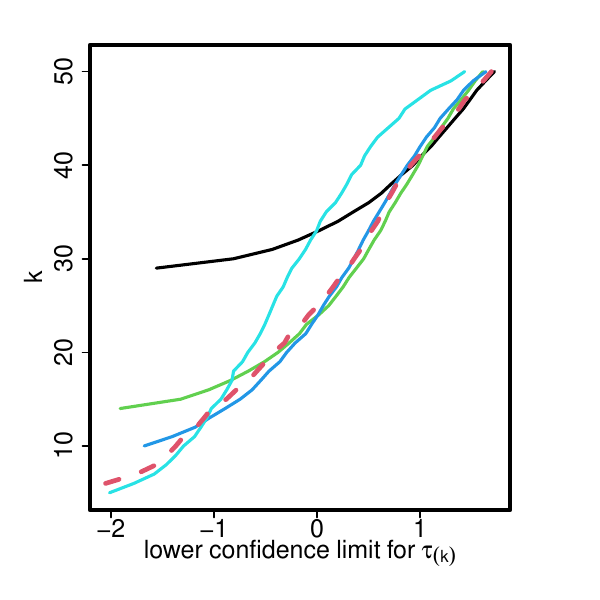}
    \caption{\small $\sigma = 1$}
    \end{subfigure}
    
    \begin{subfigure}[b]{0.32\linewidth}
    \includegraphics[width=\linewidth]{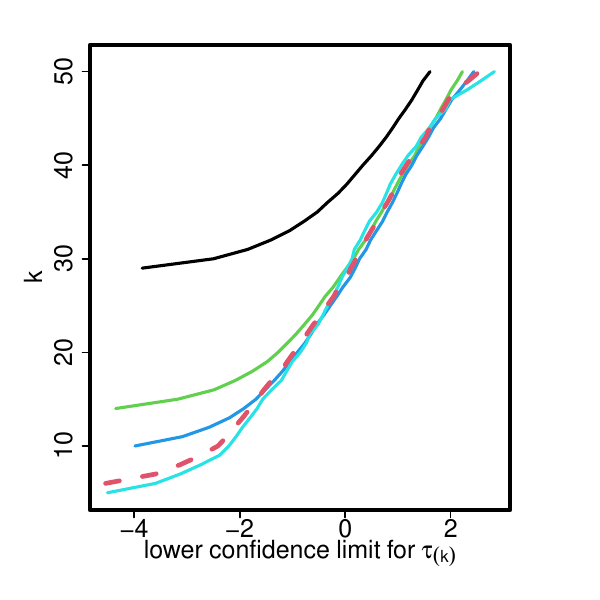}
    \caption{\small $\sigma = 2$}
    \end{subfigure}
    \begin{subfigure}[b]{0.32\linewidth}
    \includegraphics[width=\linewidth]{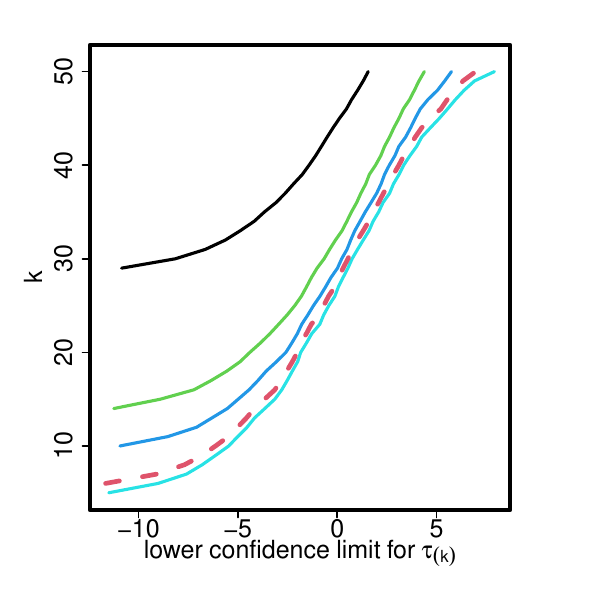}
    \caption{\small $\sigma = 5$}
    \end{subfigure} 
    \begin{subfigure}[b]{0.32\linewidth}
    \centering
    \begin{tikzpicture}
    \node[draw, rounded corners, inner sep = 2pt] {
    \begin{tikzpicture}[font = \footnotesize]
        \matrix[column sep = 1mm]{
        \draw[black, thick] (0,0) -- (0.5,0); &\node {Stephenson with $\zeta = 2$}; \\
        \draw[mygreen, thick] (0,0) -- (0.5,0); &\node {Stephenson with $\zeta = 6$}; \\
        \draw[myblue, thick] (0,0) -- (0.5,0); &\node {Stephenson with $\zeta = 10$}; \\
        \draw[mycyan, thick] (0,0) -- (0.5,0); &\node {Stephenson with $\zeta =30$}; \\
        \draw[myred, thick, dashed] (0,0) -- (0.5,0); &\node {Combined}; \\
        };
    \end{tikzpicture}
    };
    \end{tikzpicture}
    \vspace{2.25cm}
    \end{subfigure}

    \caption{Median of $90\%$ lower prediction limits of $\tau_{\treat(k)}$s over $100$ simulation replications. 
    Each solid line corresponds to results from a single Stephenson rank statistic, where the parameter for rank transformation $\zeta = 2$, $6$, $10$ and $30$, respectively. 
    The dashed line corresponds to results from our combined method proposed in Corollary \ref{cor:cre_min_pval}. 
    }
    \label{fig:cresimplot}
\end{figure}

Figure \ref{fig:cresimplot} reports the medians of the $90\%$ simultaneous lower prediction bounds for each quantile of individual treatment effects among the treated units computed over $100$ simulations, for various values of $\sigma$ in the data-generating model \eqref{eq:cresimmodel}. 
We omit the $-\infty$ lower prediction bounds for lower quantiles in the plots, and we follow this convention in all subsequent figures.
Figure \ref{fig:cresimplot} shows that the performance of the existing approach based on a single Stephenson rank sum statistic varies considerably across data-generating processes and the quantiles of interest, with different choices of $\zeta$ preferable in different scenarios. Specifically, when the treatment potential outcomes have lighter tails than the control potential outcomes (i.e., $\sigma < 1$) and larger quantiles of treatment effects are of interest, methods with smaller values of $\zeta$ yield more informative lower bounds. In contrast, when $Y_i(1)$ has heavier tails than $Y_i(0)$ (i.e., $\sigma \geq 1$), larger choices of $\zeta$ tend to perform better.

By comparison, our proposed method delivers results that are consistently close to those of the best individual rank statistic across all settings.  
For example, when $\sigma = 0.2$, the optimal choice of $\zeta$ increases as one moves from higher to lower quantiles; the combined method closely tracks this optimal choice throughout. Likewise, when $\sigma = 5$, the combined method performs nearly as well as the best individual method with $\zeta = 30$, despite incorporating the poorly performing  $\zeta = 2$ in the combination.
Overall, the combined test adapts effectively to the best performing statistic and remains robust even when some individual statistics perform poorly.

\subsection{Simulation under the SRE}\label{subsec:sre_sim}

In this subsection, we conduct simulations under the SRE to evaluate performance of methods based on the stratified rank sum statistic with the polynomial rank transformation in \eqref{eq:polynomial}, and to compare them with the proposed tests using combination methods in Sections \ref{sec:comb1} and \ref{sec:comb2}. We consider an SRE with total sample size $N = 1000$. Potential outcomes are generated from the following model:
\begin{align*}
    Y_i(0) \stackrel{\text{i.i.d.}}{\sim} \mathcal{N}(0, 1), \ Y_i(1) \stackrel{\text{i.i.d.}}{\sim} \mathcal{N}(2, \sigma^2), \ \ \text{where} \ \sigma \in \{ 0.2, \  1, \ 5 \}. 
\end{align*}
We examine three stratification scenarios:
\begin{enumerate}
\item S1 (large strata with many units): $S = 10$ strata, each with $n_s = 100$ units;
\item S2 (small strata with fewer units): $S = 100$ strata, each with $n_s = 10$ units;
\item S3 (mixed strata): $S = 56$ strata, including $33$ strata with $10$ units each, $16$ strata with $20$ units each, and $7$ strata with $50$ units each.
\end{enumerate}
Within each stratum, half of the units are randomly assigned to treatment.

In each simulation iteration, we again generate a new set of potential outcomes and a new treatment assignment vector. We then construct simultaneous prediction intervals for quantiles of the individual treatment effect among the treated units, using the weighting schemes \eqref{eq:weighting_van_opt} and \eqref{eq:weighting_van_free}. 
Specifically, we first apply a single stratified rank sum statistic using polynomial rank transformation with $\zeta = 2$, $6$ or $10$ for all strata, 
and then apply the proposed approaches in Sections \ref{sec:comb1} and \ref{sec:comb2}, combining these three polynomial rank transformations.
We evaluate their performance by comparing the median of the $90\%$ lower prediction limits across $100$ simulation replications.

\begin{figure}[htbp]
    \centering
    \begin{tikzpicture}
    \node[draw, rounded corners, inner sep = 2pt] (legend){
    \begin{tikzpicture}[font = \scriptsize]
        \matrix[column sep = 1mm]{
        \draw[black, thick] (0,0) -- (0.3,0); &\node {Polynomial with $\zeta = 2$}; &
        \draw[mygreen, thick] (0,0) -- (0.3,0); &\node {Polynomial with $\zeta = 6$}; &
        \draw[myblue, thick] (0,0) -- (0.3,0); &\node {Polynomial with $\zeta = 10$};
        & \draw[mypurple, thick, dashed] (0,0) -- (0.3,0); &\node {Comb 1}; & \draw[myred, thick, dashed] (0,0) -- (0.3,0); &\node {Comb 2}; \\
        };
    \end{tikzpicture}
    };
    \end{tikzpicture}
    \begin{subfigure}[b]{0.32\linewidth}
    \centering
    \includegraphics[width=\linewidth]{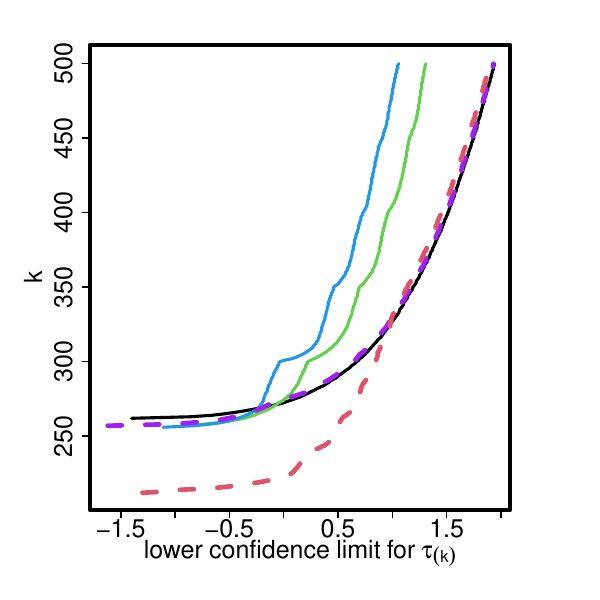}
    \caption{\small S1 and $\sigma = 0.2$
    }
    \end{subfigure} 
    \begin{subfigure}[b]{0.32\linewidth}
    \centering
    \includegraphics[width=\linewidth]{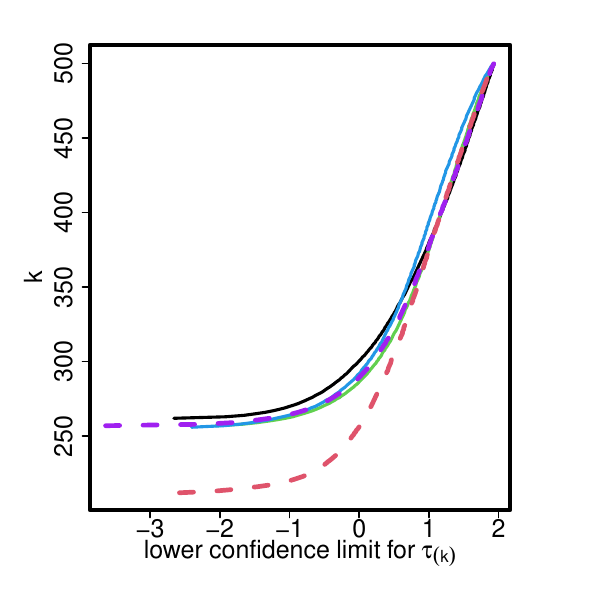}
    \caption{\small S1 and $\sigma = 1$
    }
    \end{subfigure} 
    \begin{subfigure}[b]{0.32\linewidth}
    \centering
    \includegraphics[width=\linewidth]{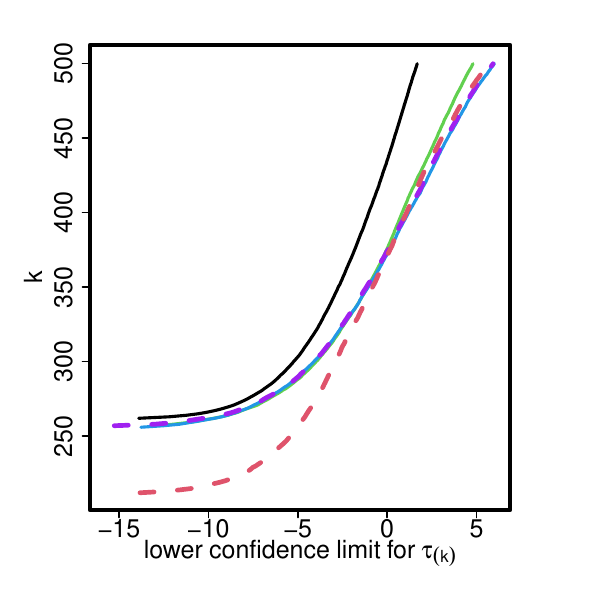}
    \caption{\small S1 and $\sigma = 5$
    }
    \end{subfigure} 

    \begin{subfigure}[b]{0.32\linewidth}
    \centering
    \includegraphics[width=\linewidth]{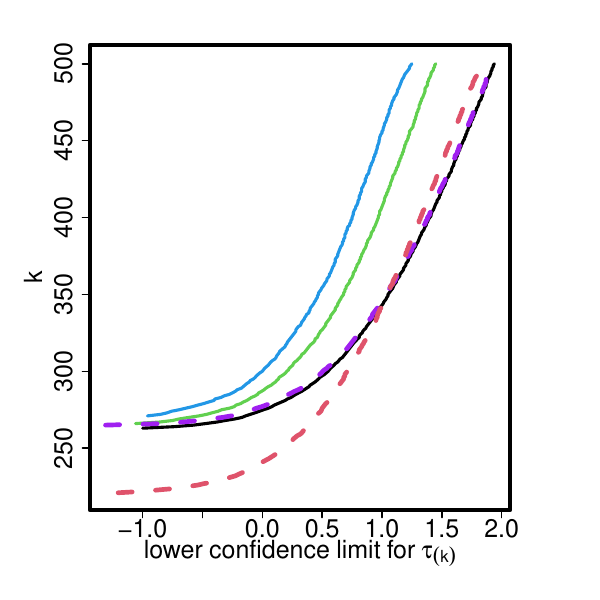}
    \caption{\small S2 and $\sigma = 0.2$
    }
    \end{subfigure} 
    \begin{subfigure}[b]{0.32\linewidth}
    \centering
    \includegraphics[width=\linewidth]{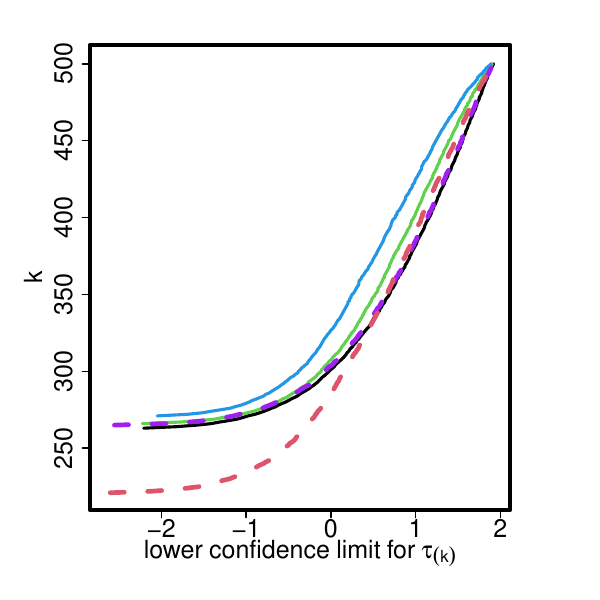}
    \caption{\small S2 and $\sigma = 1$}
    \end{subfigure} 
    \begin{subfigure}[b]{0.32\linewidth}
    \centering
    \includegraphics[width=\linewidth]{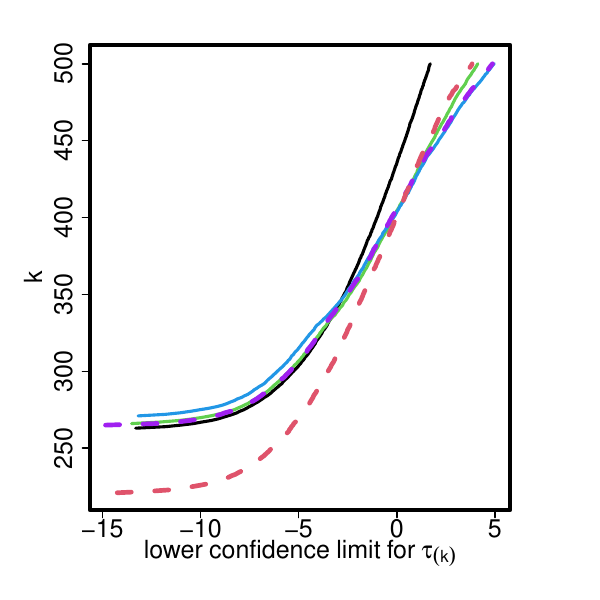}
    \caption{\small S2 and $\sigma = 5$}
    \end{subfigure} 
    
    \caption{Median of $90 \%$ lower prediction limits for treatment effect quantiles among treated units across $100$ simulation replications. 
    The first and second rows present results under stratification scenarios S1 and S2, respectively.
    Each solid line corresponds to a single stratified rank statistic using a polynomial rank transformation, with the rank transformation parameter $\zeta$ set to $2$, $6$, and $10$, respectively. 
    The dashed line corresponds to the two combined methods, where Comb 1 refers to the method in Section \ref{sec:comb1} and Comb 2 refers to the method in Section \ref{sec:comb2}.
    All results here use the weighting scheme in \eqref{eq:weighting_van_opt}, which is equivalent to the scheme in \eqref{eq:weighting_van_free} here.}
    \label{fig:sre_case1}
\end{figure}

We first summarize the simulation results under stratification scenarios S1 and S2 in Figure \ref{fig:sre_case1}. 
Under both S1 and S2, all strata contain the same numbers of treated and control units. Consequently, the two weighting schemes in \eqref{eq:weighting_van_opt} and \eqref{eq:weighting_van_free} coincide and assign equal weights to all strata.
Figure \ref{fig:sre_case1} shows that the performance of individual stratified rank statistics varies with the potential outcome distributions and the quantiles of interest.
The method that combines aggregated statistics in Section \ref{sec:comb1} performs consistently close to the optimal individual statistic across all settings.
Moreover, the method in Section \ref{sec:comb2}, which combines statistics within each stratum before aggregation, can perform even better, particularly for lower quantiles of treatment effects, with only slight losses for higher quantiles.

\begin{figure}[htbp]
    \centering
    \begin{tikzpicture}
    \node[draw, rounded corners, inner sep = 2pt] (legend){
    \begin{tikzpicture}[font = \scriptsize]
        \matrix[column sep = 1mm]{
        \draw[black, thick] (0,0) -- (0.3,0); &\node {Polynomial with $\zeta = 2$}; &
        \draw[mygreen, thick] (0,0) -- (0.3,0); &\node {Polynomial with $\zeta = 6$}; &
        \draw[myblue, thick] (0,0) -- (0.3,0); 
        &\node {Polynomial with $\zeta = 10$}; 
        &\draw[mypurple, thick, dashed] (0,0) -- (0.3,0); &\node {Comb 1}; 
        &\draw[myred, thick, dashed] (0,0) -- (0.3,0); &\node {Comb 2}; \\
        };
    \end{tikzpicture}
    };
    \end{tikzpicture}
    
    \begin{subfigure}[b]{0.33\linewidth}
    \includegraphics[width=\linewidth]{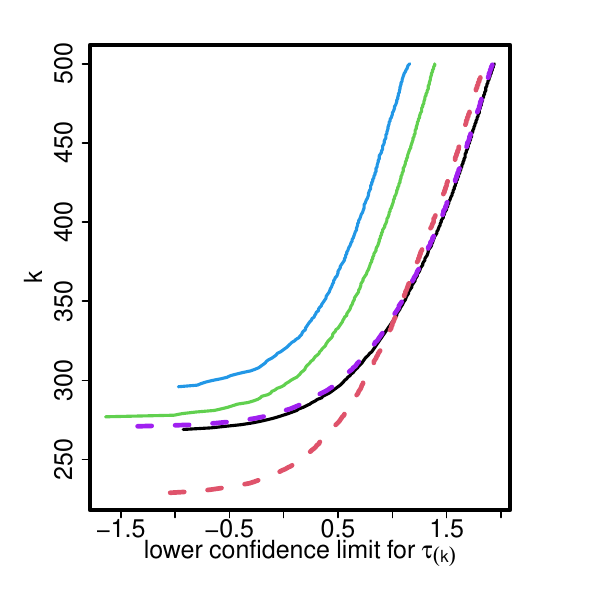}
    \caption{\footnotesize $\sigma = 0.2$, weights \eqref{eq:weighting_van_opt}}
    \end{subfigure}%
    \begin{subfigure}[b]{0.33\linewidth}
    \includegraphics[width=\linewidth]{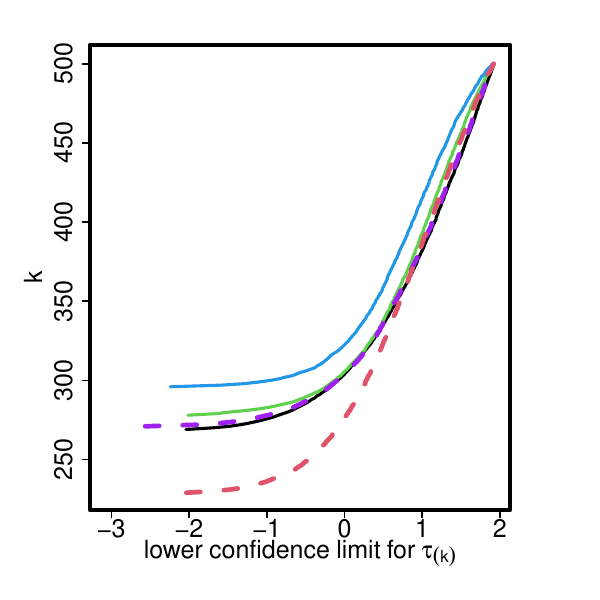}
    \caption{\footnotesize $\sigma = 1$, weights \eqref{eq:weighting_van_opt}}
    \end{subfigure}%
    \begin{subfigure}[b]{0.33\linewidth}
    \includegraphics[width=\linewidth]{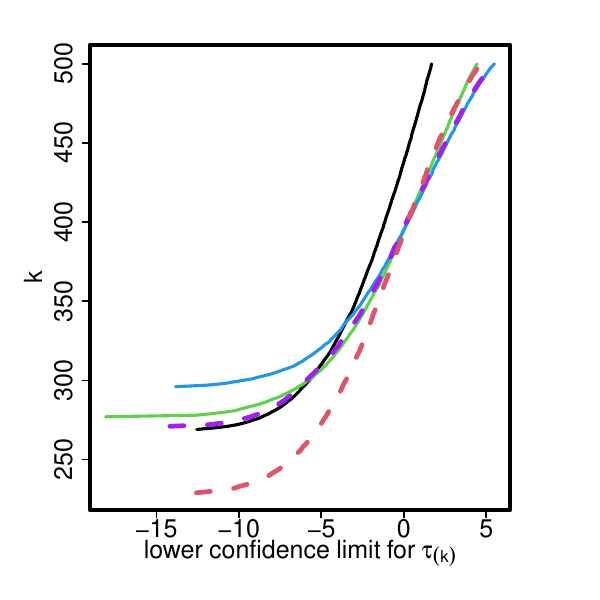}
    \caption{\footnotesize $\sigma = 5$, weights \eqref{eq:weighting_van_opt}}
    \end{subfigure} 
    
    \begin{subfigure}[b]{0.33\linewidth}
    \includegraphics[width=\linewidth]{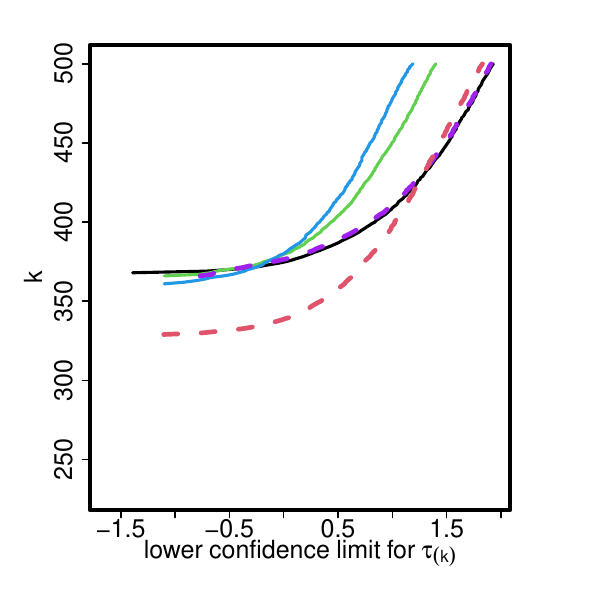}
    \caption{\footnotesize $\sigma = 0.2$, weights \eqref{eq:weighting_van_free}}
    \end{subfigure}%
    \begin{subfigure}[b]{0.33\linewidth}
    \includegraphics[width=\linewidth]{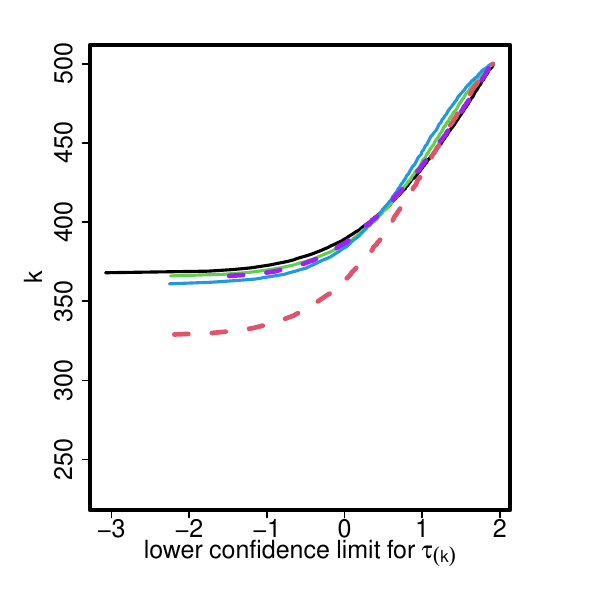}
    \caption{\footnotesize $\sigma = 1$, weights \eqref{eq:weighting_van_free}}
    \end{subfigure}%
    \begin{subfigure}[b]{0.33\linewidth}
    \includegraphics[width=\linewidth]{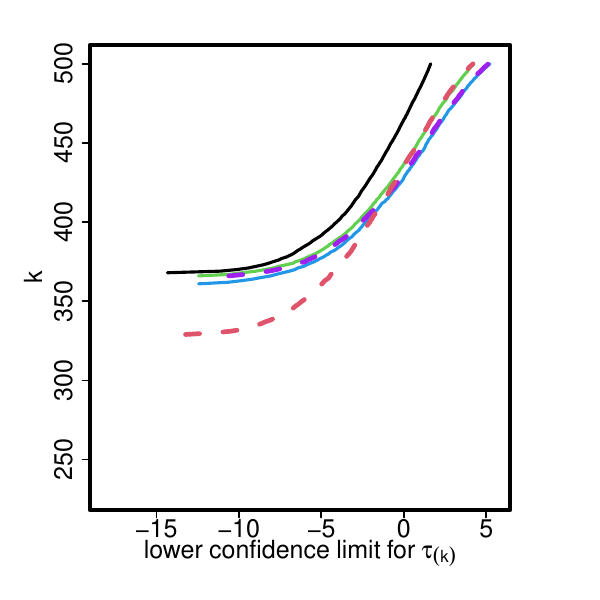}
    \caption{\footnotesize $\sigma = 5$, weights \eqref{eq:weighting_van_free}} 
    \end{subfigure}

    \caption{Median of $90 \%$ lower prediction limits for treatment effect quantiles among treated units across $100$ simulation replications, under stratification scenario S3. 
    The top three plots use weighting scheme \eqref{eq:weighting_van_opt}, 
    and the bottom three plots use weighting scheme \eqref{eq:weighting_van_free}. 
    The method associated with each line is the same as in Figure \ref{fig:sre_case1}.
    }
    
    \label{fig:sre_case2}
\end{figure}

We then consider stratification scheme S3 with mixed strata, where the
weighting schemes play a role. Figure \ref{fig:sre_case2} shows the median of
$90\%$ lower prediction limits for treatment effect quantiles among treated
units across $100$ simulation replications, using either a single stratified
rank statistic or the proposed combined approach. Comparing the two rows in
Figure \ref{fig:sre_case2}, the weighting scheme in
\eqref{eq:weighting_van_opt} has a superior performance, and the improvement
becomes more substantial as we consider lower quantiles of treatment effects.
Comparing the lower prediction limits within each plot, we can see that the
combination method in Section \ref{sec:comb1} performs consistently close to
method using the best individual stratified rank statistic, and the combination
method in Section \ref{sec:comb2} can perform even better, especially for lower
quantiles. These are consistent with the results in Figure \ref{fig:sre_case1}.

\section{Proofs of theorems and corollaries}

\subsection{Proof of Theorem \ref{thm:cre_min_pval}}

To prove Theorem \ref{thm:cre_min_pval}, we need the following lemma. 

\begin{lemma}\label{lem:dist_free_comb}
    For $1\le h \le H$,
    let
    $f^{(h)}(\cdot)$ and $\phi^{(h)}(\cdot)$ be any fixed functions,
    and $t^{(h)}(\bs{z}, \bs{y}) = \sum_{i=1}^n z_i \phi^{(h)}(\rank_i(\bs{y}))$ be a rank sum statistic,
    where the ties are broken using the first method.
    Then $\min_{1\le h \le H} f^{(h)}(t^{(h)}(\bs{Z}, \bs{y}))$ is a distribution-free statistic under the CRE, in the sense that any fixed vectors $\bs{y}, \bs{y}' \in \mathbb{R}^n$,
    $$
    \min_{1\le h \le H} f^{(h)}(t^{(h)}(\bs{Z}, \bs{y})) \sim \min_{1\le h \le H} f^{(h)}(t^{(h)}(\bs{Z}, \bs{y}'))
    $$
    where  $\bs{Z} = (Z_1, \ldots, Z_n)^\top$ is a random treatment assignment from a CRE.
\end{lemma}

\begin{proof}[\bf Proof of Lemma \ref{lem:dist_free_comb}]
    Consider any $\bs{y} \in \mathbb{R}^n$. There must exist a permutation $\pi$ such that $\rank_{\pi(i)}(\bs{y}) = i$ for $1\le i \le n$. 
    We then have, for $1\le h \le H$, 
    \begin{align*}
        t^{(h)}(\bs{z}, \bs{y}) = \sum_{i=1}^n z_i \phi^{(h)}(\rank_i(\bs{y}))
        = \sum_{i=1}^n z_{\pi(i)} \phi^{(h)}(\rank_{\pi(i)}(\bs{y}))
        = \sum_{i=1}^n z_{\pi(i)} \phi^{(h)}(i). 
    \end{align*}
    This implies that 
    \begin{align*}
        \min_{1\le h \le H} f^{(h)}(t^{(h)}(\bs{Z}, \bs{y}))
        & = 
        \min_{1\le h \le H} f^{(h)}\left( \sum_{i=1}^n Z_{\pi(i)} \phi^{(h)}(i) \right)
        \sim 
        \min_{1\le h \le H} f^{(h)}\left( \sum_{i=1}^n Z_{i} \phi^{(h)}(i) \right), 
    \end{align*}
    where the last step follows from the fact that 
    $(Z_1, \ldots, Z_n)^\top \sim (Z_{\pi(1)}, \ldots, Z_{\pi(n)})^\top$ under a CRE. 
    From the above, we know that, regardless of the values of $\bs{y}$,  $\min_{1\le h \le H} f^{(h)}(t^{(h)}(\bs{Z}, \bs{y}))$ follows the same distribution. 
    Therefore, Lemma \ref{lem:dist_free_comb} holds. 
\end{proof}

\begin{proof}[\bf Proof of Theorem \ref{thm:cre_min_pval}]
    When the null hypothesis $\mathcal{H}_{k,c}^{\treat}$ holds, 
    $\bs{\tau} \in \mathcal{H}_{k,c}^{\treat}$, and consequently
    \begin{align*}
        \underset{\delta \in \mathcal{H}_{k,c}^{\treat}}{\inf} t^{(h)}(\bs{Z}, \bs{Y} - \bs{Z} \circ \bs{\delta}) \leq t^{(h)} (\bs{Z}, \bs{Y} - \bs{Z} \circ \bs{\tau}) = t^{(h)}(\bs{Z}, \bs{Y}(0)), 
        \quad (1\le h \le H)  
    \end{align*}
    which, by the fact that $G^{(h)}(\cdot)$ is a monotone nonincreasing function, further implies that   
    \begin{align*}
        G^{(h)} \Big(\underset{\delta \in \mathcal{H}_{k,c}^{\treat}}{\inf} t^{(h)}(\bs{Z}, \bs{Y} - \bs{Z} \circ \bs{\delta}) \Big) \geq G^{(h)}(t^{(h)}(\bs{Z}, \bs{Y}(0))), 
        \quad (1\le h \le H).
    \end{align*}
    By definition, when the null hypothesis $\mathcal{H}_{k,c}^{\treat}$ holds, we then have 
    \begin{align}\label{eq:supp_min_pval_lower_bound}
        p_{k,c}^{\treat, \min} = \underset{1 \leq h \leq H}{\min} G^{(h)} \Big(\underset{\delta \in \mathcal{H}_{k,c}^{\treat}}{\inf} t^{(h)}(\bs{Z}, \bs{Y} - \bs{Z} \circ \bs{\delta}) \Big) \geq \underset{1 \leq h \leq H}{\min} G^{(h)}(t^{(h)}(\bs{Z}, \bs{Y}(0))) .  
    \end{align}
    This implies, for any $\alpha \in [0,1]$, 
    \begin{align*}
        \Pr (p_{k,c}^{\treat,\min} \leq \alpha \text{ and } \mathcal{H}_{k,c}^{\treat} \text{ holds}) &\leq \Pr \big\{ \underset{1 \leq h \leq H}{\min} G^{(h)}(t^{(h)}(\bs{Z}, \bs{Y}(0))) \leq \alpha \big\} = \Pr \big\{ \overline{t}(\bs{Z}, \bs{Y}(0)) \leq \alpha \big\}, 
    \end{align*}
    where the last equality follows from the definition of $\overline{t}(\cdot, \cdot)$. 
    From Lemma \ref{lem:dist_free_comb}, $\overline{t}(\bs{Z}, \bs{y})$ is distribution-free under the CRE. 
    We further have 
    \begin{align*}
        \Pr (p_{k,c}^{\treat,\min} \leq \alpha \text{ and } \mathcal{H}_{k,c}^{\treat} \text{ holds}) & \leq  \Pr \big\{ \overline{t}(\bs{Z}, \bs{Y}(0)) \leq \alpha \big\}
        = \Pr \big\{  \overline{t}(\bs{Z}, \bs{y}) \leq \alpha \big\} = \overline{F}(\alpha),
    \end{align*}
    where $\bs{y}$ can be any fixed vector in $\mathbb{R}^n$ and the last equality follows by definition.
    From the above, 
    we then derive Theorem \ref{thm:cre_min_pval}.
\end{proof}

\subsection{Proof of Corollary \ref{cor:cre_min_pval}}

\begin{proof}[\bf Proof of Corollary \ref{cor:cre_min_pval}]

Consider any $\alpha\in (0,1)$, and let $c_{\alpha} = \sup\{c: \overline{F}(c) \le \alpha\}$. 
Below we consider two cases. 

We first consider the case where $\overline{F}(c_{\alpha}) \le \alpha$. We then have, for any $c\in \mathbb{R}$, $\overline{F}(c) \le \alpha$ if and only if $c\le c_{\alpha}$. Consequently, 
\begin{align*}
    \Pr\{ 
    \overline{F}(p_{k,c}^{\treat,\min}) \le \alpha 
    \text{ and } \mathcal{H}_{k,c}^{\treat} \text{ holds}
    \}
    = 
    \Pr( 
    p_{k,c}^{\treat,\min} \le c_\alpha 
    \text{ and } \mathcal{H}_{k,c}^{\treat} \text{ holds}
    )
    \le 
    \overline{F}(c_{\alpha}) \le \alpha, 
\end{align*}
where the second last inequality follows from Theorem \ref{thm:cre_min_pval}. 

We then consider the case where $\overline{F}(c_{\alpha}) > \alpha$.
We then have, for any $c\in \mathbb{R}$, $\overline{F}(c) \le \alpha$ if and only if $c < c_{\alpha}$. 
Consequently, 
\begin{align*}
    & \quad \ \Pr\{ 
    \overline{F}(p_{k,c}^{\treat,\min}) \le \alpha 
    \text{ and } \mathcal{H}_{k,c}^{\treat} \text{ holds}
    \}
    \\
    & = 
    \Pr( 
    p_{k,c}^{\treat,\min} < c_\alpha 
    \text{ and } \mathcal{H}_{k,c}^{\treat} \text{ holds}
    )
    = 
    \lim_{m\rightarrow \infty}
    \Pr( 
    p_{k,c}^{\treat,\min} \le c_\alpha - 1/m 
    \text{ and } \mathcal{H}_{k,c}^{\treat} \text{ holds}
    )
    \\
    & \le 
    \lim_{m\rightarrow \infty}
    \overline{F}(c_{\alpha}-1/m)
    &  \le \alpha, 
\end{align*}
where the second last inequality follows from Theorem \ref{thm:cre_min_pval}. 

From the above, $\overline{F}(p_{k,c}^{\treat,\min})$ is a valid $p$-value for testing the null hypothesis $\mathcal{H}_{k,c}^{\treat}$, i.e., Corollary \ref{cor:cre_min_pval} holds. 
\end{proof}

\subsection{Proof of Theorem \ref{thm:cre_min_pval_combined}}

\begin{proof}[\bf Proof of Theorem \ref{thm:cre_min_pval_combined}]

From \eqref{eq:supp_min_pval_lower_bound}, 
\begin{align*}
        p_{k,c}^{\treat, \min} & = \underset{1 \leq h \leq H}{\min} G^{(h)} \Big(\underset{\delta \in \mathcal{H}_{k,c}^{\treat}}{\inf} t^{(h)}(\bs{Z}, \bs{Y} - \bs{Z} \circ \bs{\delta}) \Big)   
        = 
        \underset{1 \leq h \leq H}{\min} G^{(h)} \Big( t^{(h)}( \bs{Z}, \bs{Y} - \bs{Z} \circ \bs{\zeta}_{k, c}^\treat ) \Big)
        \\
        & = 
        \overline{t}(\bs{Z}, \bs{Y} - \bs{Z} \circ \bs{\zeta}_{k, c}^\treat), 
\end{align*}
where the second last equality follows from the property of each individual rank sum statistic as studied in \citet{CDLM21quantile} and \citet{ZL24quantile}, 
and the last equality follows by definition. 
Importantly, the infimum of each individual rank sum statistic over the set $\mathcal{H}_{k,c}^{\treat}$ is achieved at the same hypothesized individual treatment effect vector $\xi_{k,c}^{\treat}$.
In addition, for any $\bs{\delta} \in \mathcal{H}_{k,c}^{\treat}$, we have 
\begin{align*}
    \overline{t}(\bs{Z}, \bs{Y} - \bs{Z} \circ \bs{\delta})
    & = 
    \underset{1 \leq h \leq H}{\min} G^{(h)} \Big( t^{(h)}( \bs{Z}, \bs{Y} - \bs{Z} \circ \bs{\delta} ) \Big)
    \\
    & \le 
    \underset{1 \leq h \leq H}{\min} G^{(h)} \Big(\underset{\delta \in \mathcal{H}_{k,c}^{\treat}}{\inf} t^{(h)}(\bs{Z}, \bs{Y} - \bs{Z} \circ \bs{\delta}) \Big)   
    =p_{k,c}^{\treat, \min}, 
\end{align*}
where the first equality follows by definition, and the inequality follows by the fact that $G^{(h)} (\cdot)$ is a nonincreasing function. 
These imply that 
\begin{align*}
    \sup_{\bs{\delta} \in \mathcal{H}_{k,c}^{\treat}} \overline{t}(\bs{Z}, \bs{Y} - \bs{Z} \circ \bs{\delta}) \le  p_{k,c}^{\treat, \min}
    = 
    \overline{t}(\bs{Z}, \bs{Y} - \bs{Z} \circ \bs{\zeta}_{k, c}^\treat)
    \le \sup_{\bs{\delta} \in \mathcal{H}_{k,c}^{\treat}} \overline{t}(\bs{Z}, \bs{Y} - \bs{Z} \circ \bs{\delta}). 
\end{align*}
Thus, the quantities in the above equation must be all equal.

By definition, we know that, for $c\in \mathbb{R}$, 
\begin{align*}
    \overline{F}(c) = 
    \Pr \{
    \overline{t}(\bs{Z}, \bs{y}) \le c
    \}
    = 
    \Pr \{
    - \overline{t}(\bs{Z}, \bs{y}) \ge - c
    \}
    = 
    G^{(-\overline{t})} (-c). 
\end{align*}
We then have 
\begin{align*}
    \overline{F}(p_{k,c}^{\treat, \min}) &= 
    G^{(-\overline{t})} ( - p_{k,c}^{\treat, \min} )
    = 
    G^{(-\overline{t})} \big( - \overline{t}(\bs{Z}, \bs{Y} - \bs{Z} \circ \bs{\zeta}_{k, c}^\treat) \big)
    \\
    & = 
    G^{(-\overline{t})} \big( - \sup_{\bs{\delta} \in \mathcal{H}_{k,c}^{\treat}} \overline{t}(\bs{Z}, \bs{Y} - \bs{Z} \circ \bs{\delta})  \big)\\
    & = 
    G^{(-\overline{t})} \big( \inf_{\bs{\delta} \in \mathcal{H}_{k,c}^{\treat}} -\overline{t}(\bs{Z}, \bs{Y} - \bs{Z} \circ \bs{\delta})  \big).
\end{align*}
Therefore, Theorem \ref{thm:cre_min_pval_combined} holds. 
\end{proof}

\subsection{Proof of Theorem \ref{thm:cre_min_pval_combined_gaussian}}

\begin{proof}[\bf Proof of Theorem \ref{thm:cre_min_pval_combined_gaussian}]
First, from Lemma \ref{lem:dist_free_comb}, 
\begin{align*}
    \tilde{t}(\bs{z}, \bs{y}) = \max_{1\le h\le H} \frac{t^{(h)}(\bs{z}, \bs{y})-\mu^{(h)}}{\sigma^{(h)}}
    = 
    - \min_{1\le h\le H} \left\{  - \frac{t^{(h)}(\bs{z}, \bs{y})-\mu^{(h)}}{\sigma^{(h)}} \right\}
\end{align*}
is distribution-free under the CRE.

Second, we prove the validity of the $p$-value $p_{k,c}^{\treat (\Tilde{t})}$, which follows by the same logic as Theorem 1 in \citet{ZL24quantile}. 
If the null $\mathcal{H}_{k,c}^{\treat}$ holds, then $\bs{\tau} \in \mathcal{H}_{k,c}$, and 
\begin{align*}
    \underset{\delta \in \mathcal{H}_{k,c}^{\treat}}{\inf} \Tilde{t}(\bs{Z}, \bs{Y} - \bs{Z} \circ \bs{\delta}) \leq \Tilde{t}(\bs{Z}, \bs{Y} - \bs{Z} \circ \bs{\tau}) = \Tilde{t}(\bs{Z}, \bs{Y}(0)). 
\end{align*}
Because $G^{(\Tilde{t})}(\cdot)$ is a monotone nonincreasing function, we then have, under the null  $\mathcal{H}_{k,c}^{\treat}$,  
\begin{align*}
    p_{k,c}^{\treat (\Tilde{t})} \equiv G^{(\Tilde{t})}(\underset{\delta \in \mathcal{H}_{k,c}^{\treat}}{\inf} \Tilde{t}(\bs{Z}, \bs{Y} - \bs{Z} \circ \bs{\delta})) \geq G^{(\Tilde{t})} (\Tilde{t}(\bs{Z}, \bs{Y}(0)) ). 
\end{align*}
Consequently, 
\begin{align*}
   \Pr(p_{k,c}^{\treat (\Tilde{t})} \leq \alpha  \text{ and } \mathcal{H}_{k,c}^{\treat} \text{ holds} )  \leq \Pr \Big( G^{(\Tilde{t})} (\Tilde{t}(\bs{Z}, \bs{Y}(0)) ) \leq \alpha \Big) \leq \alpha, 
\end{align*}
where the last inequality uses the fact that $G^{(\Tilde{t})} ( \cdot)$ is the tail probability of $\Tilde{t}(\bs{Z}, \bs{Y}(0))$ and Lemma A4 in \citet{CDLM21quantile}. 

Finally, we prove the equivalent form in \eqref{eq:pval_combined_minpval_gaussian}. 
From \citet{CDLM21quantile} and \citet{ZL24quantile} and as discussed in the proof of Theorem \ref{thm:cre_min_pval_combined}, 
each individual rank sum statistic $\Tilde{t}(\bs{Z}, \bs{Y} - \bs{Z} \circ \bs{\delta} )$ achieves its infimum over $\delta \in \mathcal{H}_{k,c}^{\treat}$ at the same value $\bs{\xi}_{k,c}^{\treat}$. 
Thus, we must have 
\begin{align*}
     \underset{\delta \in \mathcal{H}_{k,c}^{\treat}}{\inf} \Tilde{t}(\bs{Z}, \bs{Y} - \bs{Z} \circ \bs{\delta} )
     = 
     \Tilde{t}(\bs{Z}, \bs{Y} - \bs{Z} \circ \bs{\xi}_{k,c}^{\treat} ),
\end{align*}
which immediately implies \eqref{eq:pval_combined_minpval_gaussian}.

From the above, Theorem \ref{thm:cre_min_pval_combined_gaussian} holds. 
\end{proof}

\subsection{Proof of Corollary \ref{cor:interval_cre}}

\begin{proof}[\bf Proof of Corollary \ref{cor:interval_cre}]
    The proof of Corollary \ref{cor:interval_cre} follows by the same logic as the proof of Theorems 1 and 2 in \citet{ZL24quantile}, and is thus omitted here for conciseness.  
\end{proof}

\subsection{Proof of Theorem \ref{thm:combined_stat_max_sre}}

To prove Theorem \ref{thm:combined_stat_max_sre}, we need the following lemma. 

\begin{lemma}\label{lem:dist_free_comb_sre}
    For $1\le h \le H$,
    let $$t^{(h)}(\bs{z}, \bs{y}) = \sum_{s = 1}^S t_s^{(h)}(\bs{z}_s, \bs{y}_s) = \sum_{s=1}^S \sum_{i=1}^{n_s} z_{si} \phi_s^{(h)}(\rank_i(\bs{y}_s))$$ be a stratified rank sum statistic defined as in \eqref{eq:rank_sum_stratified},
    where the ties are broken using the first method
    and $\phi_s^{(h)}(\cdot)$s are fixed functions,
    and let $f^{(h)}(\cdot)$ be any fixed function.
    Then $\max_{1\le h \le H}$ $f^{(h)}(t^{(h)}(\bs{Z}, \bs{y}))$ is a distribution-free statistic under the SRE, in the sense that any fixed vectors $\bs{y}, \bs{y}' \in \mathbb{R}^n$,
    $$
    \max_{1\le h \le H} f^{(h)}(t^{(h)}(\bs{Z}, \bs{y})) \sim \max_{1\le h \le H} f^{(h)}(t^{(h)}(\bs{Z}, \bs{y}'))
    $$
    where  $\bs{Z} = (Z_1, \ldots, Z_n)^\top$ is a random treatment assignment from a SRE.
\end{lemma}

\begin{proof}[\bf Proof of Lemma \ref{lem:dist_free_comb_sre}]
    Consider any $\bs{y} \in \mathbb{R}^n$. 
    For each stratum $s$, there must exist a permutation $\pi_s$ of $(1,2,\ldots, n_s)$ such that    
    $\rank_{\pi_s(i)}(\bs{y}_s) = i$ for $1\le i \le n_s$. 
    We then have, for $1\le h \le H$ and $1\le s \le S$, 
    \begin{align*}
        t_s^{(h)}(\bs{z}_s, \bs{y}_s) = \sum_{i=1}^{n_s} z_{si} \phi_s^{(h)}(\rank_i(\bs{y}_s))
        =
        \sum_{i=1}^{n_s} z_{s\pi_s(i)} \phi_s^{(h)}(\rank_{\pi_s(i)}(\bs{y}_s))
        = \sum_{i=1}^{n_s} z_{s\pi_s(i)} \phi_s^{(h)}(i). 
    \end{align*}
    This then implies that
    \begin{align*}
        & \quad \ \max_{1\le h \le H} f^{(h)}(t^{(h)}(\bs{Z}, \bs{y}))
        \\
        & = 
        \max_{1\le h \le H} f^{(h)}\Big(
        \sum_{s = 1}^S t_s^{(h)}(\bs{Z}_s, \bs{y}_s)
        \Big)
        =
        \max_{1\le h \le H} f^{(h)}\Big(
        \sum_{s=1}^S \sum_{i=1}^{n_s} Z_{s\pi_s(i)} \phi_s^{(h)}(i)
        \Big)\\
        & \sim 
        \max_{1\le h \le H} f^{(h)}\Big(
        \sum_{s=1}^S \sum_{i=1}^{n_s} Z_{si} \phi_s^{(h)}(i)
        \Big), 
    \end{align*}
    where the last step follows from the fact that 
    \begin{align*}
        (Z_{11}, \ldots, Z_{1n_1}, \ldots, Z_{S1}, \ldots, Z_{Sn_S})
        \sim 
        (Z_{1\pi_1(1)}, \ldots, Z_{1\pi_1(n_1)}, \ldots, Z_{S\pi_S(1)}, \ldots, Z_{S\pi_S(n_S)})
    \end{align*}
    under the SRE. 
    From the above, regardless of the values of $\bs{y}$,  $\max_{1\le h \le H} f^{(h)}(t^{(h)}(\bs{Z}, \bs{y}))$ follows the same distribution. 
    Therefore, Lemma \ref{lem:dist_free_comb_sre} holds. 
\end{proof}

\begin{proof}[\bf Proof of Theorem \ref{thm:combined_stat_max_sre}]

From Lemma \ref{lem:dist_free_comb_sre}, $\tilde{t}(\bs{Z}, \bs{y})$ is distribution-free under the SRE. 
By the same logic as Theorem 6 in \citet{ZL24quantile}, 
\begin{align*}
    p_{k,c}^{1 (\tilde{t})} \equiv G^{(\tilde{t})} \Big( \underset{\delta \in \mathcal{H}_{k,c}^{\treat}}{\inf} \tilde{t}(\bs{Z}, \bs{Y} - \bs{Z} \circ \bs{\delta} )\Big)
\end{align*}
is a valid $p$-value for testing the null hypothesis $\mathcal{H}_{k,c}^{\treat}$. 
Therefore, we derive Theorem \ref{thm:combined_stat_max_sre}. 
\end{proof}

\begin{proof}[\bf Proof of Theorem \ref{thm:combine_then_aggregate}] 
From Lemma \ref{lem:dist_free_comb}, 
the statistic in \eqref{eq:comb_per_stratum} for each stratum $s$ is distribution-free under the complete randomization within the stratum, in the sense that $\tilde{t}_s(\bs{Z}_s, \bs{y}_s) \sim \tilde{t}_s(\bs{Z}_s, \bs{y}'_s)$ for any constant $\bs{y}_s$ and $\bs{y}'_s$. 
By the mutual independence of treatment assignments across strata, we can know that the statistic in \eqref{eq:comb_then_agg} is distribution-free under the SRE, in the sense that $\dtilde{t}(\bs{Z}, \bs{y}) \sim \dtilde{t}(\bs{Z}, \bs{y}')$. 
Theorem \ref{thm:combine_then_aggregate} then follows by the same logic as Theorem 6 in \citet{ZL24quantile}.     
\end{proof}

\end{document}